\documentclass{article}

\usepackage{arxiv}

\usepackage[utf8]{inputenc} 
\usepackage[T1]{fontenc}    
\usepackage{hyperref}       
\usepackage{url}            
\usepackage{booktabs}       
\usepackage{amsfonts}       
\usepackage{nicefrac}       
\usepackage{microtype}      
\usepackage{lipsum}		
\usepackage{graphicx}
\usepackage{natbib}
\usepackage{doi}

\usepackage{amsmath}
\usepackage{dsfont}
\usepackage{booktabs}
\usepackage{longtable}
\usepackage{multirow}
\usepackage{graphicx}
\usepackage[utf8]{inputenc}
\usepackage[english]{babel}
\usepackage{enumitem}
\usepackage{amsthm}
\usepackage[ruled,vlined]{algorithm2e}
\usepackage{accents}
\usepackage{graphicx}
\usepackage{blindtext}
\usepackage{csvsimple}
\usepackage{tikz}
\usetikzlibrary{shapes}
\usepackage{pgfplots}
	\pgfplotsset{compat=newest}
\usepackage{pgfplotstable}
\usepgfplotslibrary{fillbetween}
\usepackage{amsfonts}

\usepackage{bm}
\usepackage{mathtools}
\usepackage[utf8]{inputenc} 
\usepackage[T1]{fontenc}    
\usepackage{hyperref}       
\usepackage{url}            
\usepackage{booktabs}       
\usepackage{amsfonts}       
\usepackage{nicefrac}       
\usepackage{microtype}      
\usepackage{lipsum}		
\usepackage{graphicx}
\usepackage{doi}
\newcommand{\fig}[4]{\begin{figure}[ht]\centering\includegraphics[width=#1\linewidth]{Figures/#2}\caption{#3}\label{#4}\end{figure}\\}
\usepackage{makecell}
    
\newtheorem{theorem}{Theorem}[section]

\newtheorem{proposition}{Proposition}[section]
\newtheorem{lemma}[theorem]{Lemma}

\theoremstyle{definition}
\newtheorem{definition}{Definition}[section]

\theoremstyle{remark}
\newtheorem*{remark}{Remark}
\newcommand\R[0]{\mathbb{R}}
\newcommand\argmax[0]{\mathrm{argmax}}
\newcommand\argmin[0]{\mathrm{argmin}}
\newcommand\Cobs[0]{\mathbf{C}_{\mathrm{obs}}}
\newcommand\Cn[0]{\mathbf{C}_n}
\newcommand\Ctot[0]{\boldsymbol\Sigma_n}
\newcommand\Ctotup[0]{\boldsymbol\Sigma_{n+1}}
\newcommand\Ceff[0]{\mathbf{C}_{\mathrm{eff}}}
\newcommand\Cnu[0]{\mathbf{C}_\nu}
\newcommand\xmap[0]{x_n^{(m)}}

\title{Sequential design for surrogate modeling in Bayesian inverse problems}

\author{\hspace{1mm}Paul Lartaud \\
        CEA DAM Île-de-France\\
        Centre de Mathématiques Appliquées, \\
        École polytechnique, \\
        Institut Polytechnique de Paris\\
	\texttt{paul.lartaud@polytechnique.edu} \\
	\And \hspace{1mm}Philippe Humbert \\
        CEA DAM Île-de-France\\
	\And \hspace{1mm}Josselin Garnier \\
        Centre de Mathématiques Appliquées, \\
        École polytechnique, \\
        Institut Polytechnique de Paris\\
}

\date{}

\renewcommand{\headeright}{}
\renewcommand{\undertitle}{}
\renewcommand{\shorttitle}{Sequential design for Bayesian inverse problems with surrogate modeling}

\hypersetup{
pdftitle={Sequential design - Lartaud},
pdfsubject={q-bio.NC, q-bio.QM},
pdfauthor={P. Lartaud, P. Humbert and J. Garnier},
pdfkeywords={sequential design, surrogate model, inverse problem},
}

\begin{document}
\maketitle

\begin{abstract}
Sequential design is a highly active field of research in active learning which provides a general framework for designing computer experiments with limited computational budgets. It aims to create efficient surrogate models to replace complex computer codes. Some sequential design strategies can be understood within the Stepwise Uncertainty Reduction (SUR) framework. In the SUR framework, each new design point is chosen by minimizing the expectation of a metric of uncertainty with respect to the yet unknown new data point. These methods offer an accessible framework for sequential experiment design, including almost sure convergence for common uncertainty functionals. 

This paper introduces two strategies. The first one, entitled Constraint Set Query (CSQ) is adapted from D-optimal designs where the search space is constrained in a ball for the Mahalanobis distance around the maximum a posteriori. The second, known as the IP-SUR (Inverse Problem SUR) strategy, uses a weighted-integrated mean squared prediction error as the uncertainty metric and is derived from SUR methods. It is tractable for Gaussian process surrogates with continuous sample paths. It comes with a theoretical guarantee for the almost sure convergence of the uncertainty functional. The premises of this work are highlighted in various test cases, in which these two strategies are compared to other sequential designs.
\end{abstract}

\section{Introduction}\label{sec:intro}
In the field of uncertainty quantification, the resolution of ill-posed inverse problems with a Bayesian approach is a widely used method. It gives access to the posterior distribution of the uncertain inputs from which various quantities of interest can be calculated for decision-making \citep{stuart2010inverse, kaipio2006statistical}. However, sampling the posterior distribution often requires a large number of calls to the direct model. This limits the tractability of such methods when the direct model is a complex computer code. To overcome this obstacle, it is common practice to replace the direct model with a surrogate such as deep neural networks \citep{yan2019adaptive}, physics-informed neural networks \citep{li2023surrogate}, polynomial chaos \citep{yan2019adaptivepce} or Gaussian processes \citep{drovandi2018accelerating, frangos2010surrogate, sauer2023active, chen2021gaussian}. These surrogates help in reducing the computational burden of the inverse problem resolution. The quality of the uncertainty prediction is then directly linked to the quality of the surrogate model. In some applications \citep{mai2017seismic}, the computational budget implies a scarcity of the numerical data used to build the surrogate model, since they are given by calls of the direct model. To build the best possible surrogate model, adaptive strategies have been developed to make the most out of the given computational budget \citep{wang2018adaptive, wan2016stochastic}. 

This is the premise of sequential designs for computer experiments \citep{santner2003design, gramacy2009adaptive, sacks1989designs}. While space-filling designs based on Latin Hypercube Sampling (LHS) \citep{stein1987large} or minimax distance designs \citep{johnson1990minimax} have been widely investigated, they are not suited for the specific problem at stake where the region of interest only covers a small fraction of the input space. For such cases, criterion-based designs have been introduced in various domains, such as reliability analysis \citep{lee2008sampling, du2002sequential, agrell2021sequential, azzimonti2021adaptive, dubourg2013metamodel, cole2023entropy}, Bayesian optimization \citep{shahriari2015taking, imani2020bayesian}, contour identification \citep{ranjan2008sequential} or more recently in Bayesian calibration \citep{ezzat2018sequential, surer2024sequential, kennedy2001bayesian} or in multi-fidelity surrogate modeling \citep{stroh2022sequential}. Most common approaches include criteria based on the maximization of the information gain \citep{shewry1987maximum}, D-optimal criteria based on the maximization of the predictive covariance determinant \citep{wang2006d, de1995d, osio1996engineering, mitchell2000algorithm} or designs minimizing a functional of the mean squared prediction error (MSPE) \citep{kleijnen2004application}. Other applications use the trace of the covariance instead of the determinant. For our applicative case, the strong correlations across outputs point toward the use of the determinant. Weighted integrated MSPE optimal designs have been used to improve surrogate models and outperform traditional LHS sampling \citep{picheny2010adaptive}. In Bayesian inverse problems, sequential design strategies have been investigated to produce estimators for the inverse problem likelihood \citep{sinsbeck2021sequential} or the Bayesian model evidence \citep{sinsbeck2017sequential} to provide insight on model selection. In \citep{li2014adaptive}, the design points for the surrogate model are obtained iteratively from variational posterior distributions designed to minimize the Kullback-Leibler divergence with the true posterior. 

In this paper, two novel sequential design strategies named Inverse Problem SUR (IP-SUR) and Constraint Set Query (CSQ) are introduced for Gaussian process surrogate models in Bayesian inverse problems. The first one is an objective-oriented design based on the Stepwise Uncertainty Reduction (SUR) paradigm applied to the posterior-weighted integrated MSPE. It is shown to be tractable for GP surrogates. Its integration into the SUR paradigm makes it easily interpretable. Moreover, it is supported by a theoretical guarantee of almost sure convergence of the metric of uncertainty, which is rarely obtained for sequential design strategies. The CSQ design is a simpler strategy adapted from D-optimal design strategies. It provides a more accessible design though it lacks the convergence guarantee of IP-SUR. Overall, the two sequential design strategies presented in this work are complementary. CSQ appears more computationally efficient for higher dimensions because it circumvents the need for intermediate MCMC samplings, whereas IP-SUR is more theoretically sound. We insist that the proposed strategies be used for fine-tuning a surrogate to a given inverse problem. They should not be used for the initial construction of a surrogate. For this task, space-filling or criterion-based designs, such as D-optimal designs, are more relevant.

The following section presents the fundamentals of Gaussian process regression, Bayesian inverse problems, and SUR designs. Then, we present our two methods and develop the theoretical foundations of the IP-SUR strategy. Finally, we apply our strategies to several test cases and compare them to other designs (naive, D-optimal, Bayes risk minimization).

\section{Bayesian inverse problems}\label{sec:inv_problem_section}
\subsection{Problem statement}
Consider the following inverse problem. We would like to identify the parameters $x \in \mathcal{X} \subset \R^p$ based on the observations of some vectorial quantity $y \in  \R^d$. The link between the inputs $x$ and the outputs $y$ is provided by a direct model $f \colon \mathcal{X} \xrightarrow[]{} \R^d$. Assume that we are given some noisy observations of the direct model $\mathbf{y} = (y^{(k)})_{1 \leq k \leq N}$ for $N \geq 1$. We would like to estimate parameters $x$ and quantify the uncertainty associated with this estimation. The traditional approach is to solve the inverse problem with a Bayesian approach \citep{dashti2013bayesian, scales2001prior}. 

In the Bayesian paradigm, the inputs $x$ and the outputs $y$ are considered random variables and the goal is to estimate the posterior distribution $p(\ \cdot \ | \mathbf{y})$, which is the probability distribution of the inputs $x$ conditioned by the noisy observations $\mathbf{y}$ of the outputs. We assume the observations are independent and identically distributed with zero-mean Gaussian noise:
\begin{equation}\label{eq:observation_def}
    y^{(k)} = f(x) + \varepsilon^{(k)} \text{ with } \varepsilon^{(k)} \sim \mathcal{N}\left(\mathbf{0}, \Cobs \right)
\end{equation}
for $1 \leq k  \leq N$, and where $\Cobs \in S_d(\R)^+$ is the covariance matrix of the observations, and $S_d(\R)^+$ is the set of symmetric positive-definite matrices of size $d \times d$.

Let $x \in \mathcal{X}$. The posterior distribution density $p(x | \mathbf{y})$ can be expressed with Bayes' theorem as the product of a prior distribution with density $p(x)$ and an analytically tractable likelihood $L(\mathbf{y}|x)$:
\begin{equation*}
    p(x | \mathbf{y}) \propto p(x) L(\mathbf{y} | x) \propto p(x) \exp\left(- \frac{1}{2}\sum\limits_{k=1}^N \| y^{(k)} - f(x)\|^2_{\Cobs} \right)
\end{equation*}
where the notation $\| a \|^2_{\mathbf{C}} = a^T \mathbf{C}^{-1} a$ refers to the Mahanalobis norm, for any vector $a \in \R^d$ and $\mathbf{C} \in S_d(\R)^+$.

The prior $p(\ \cdot \ )$ is chosen by the user directly. It can incorporate expert knowledge or information from other observations. It can also be taken mostly non-informative \citep{jeffreys1946invariant, ghosh2011objective, consonni2018prior}. 

Given a non-normalized posterior density, the posterior distribution can be sampled with Monte Carlo Markov Chain (MCMC) methods. MCMC algorithms produce ergodic Markov chains whose invariant distribution is the target posterior distribution. MCMC methods remove the need to evaluate the multiplicative constant $\int p(x') L(\mathbf{y} | x') dx'$ which is often intractable.

The Bayesian approach as depicted here has two major flaws, however. The direct model $f$ is assumed to be known, and the sampling of the posterior distribution requires a large number of calls to this direct model. For real-world applications, the direct model is often given by experiments or by complex computer codes. In some cases, an analytical direct model can be found given strong assumptions, but it comes with a systematic bias that needs to be accounted for. Thus, obtaining the posterior distribution is often too computationally expensive. A way to circumvent this obstacle is to rely on surrogate models.

\subsection{Gaussian process surrogate models}\label{subsec:sm_with_inv_problem}
In this work, we consider Gaussian Process (GP) based surrogate models because GPs induce closed-form posterior distributions. We briefly recall the basic concepts of GP surrogates. For a comprehensive review on the subject, we refer to \citep{williams2006gaussian, gramacy2020surrogates}. 

Let $\mathcal{X} \subset \R^p$ be an index set. A Gaussian process $z$ is a stochastic process such that for any finite collection $\mathbf{x} = (x_1, ..., x_n) \in \mathcal{X}^n$ with $n \geq 1$, the random variable $z(\mathbf{x}) = (z(x_1), ..., z(x_n))$ follows a multivariate normal distribution.  

The distribution of a Gaussian process is entirely determined by its mean function $m\colon \mathcal{X} \longrightarrow \R$ and its positive definite covariance function $k\colon \mathcal{X} \times \mathcal{X} \longrightarrow \R$. Throughout this work, a GP with mean function $m$ and covariance function $k$ is denoted by  $z \sim \mathcal{GP}\left(m(x), k(x, x') \right)$. 

In this work, we are considering multi-output GP models. These models are built with the Linear Model of Coregionalization (LMC) introduced in \citep{bonilla2007multi}. LMC models allow correlations across outputs to be modeled with a linear combination of independent latent GPs. They produce a multi-output covariance function $k \colon \mathcal{X} \times \mathcal{X} \rightarrow \mathcal{M}_d(\R)$ where $\mathcal{M}_d(\R)$ is the set of real-valued square matrices of size $d$.

Consider a collection of inputs $\mathbf{x} \in \mathcal{X}^n$ and a multi-output GP $z \sim \mathcal{GP}(\mathbf{0}, k(x, x'))$.  The distribution of the GP for the inputs $\mathbf{x}$ is written as:
\begin{equation*}
    z(\mathbf{x}) \sim \mathcal{N}\left( \mathbf{0}, \mathbf{K}(\mathbf{x}) \right)
\end{equation*}
where $\mathbf{K}(\mathbf{x})$ is a $nd \times nd$ block matrix whose constitutive blocks are given by $k(x_i, x_j)$ for $1 \leq i, j \leq n$.  

Gaussian processes can be used to introduce a prior on a functional space. The covariance function defines the class of random functions represented by the GP prior. The most standard covariance kernels are the stationary RBF and Matérn kernels. Depending on the expected properties of the target function (periodicity, regularity, ...) various classes of covariance kernels can be used. The mean function is more straightforward and is often set to a constant, though more advanced approaches exist to improve the performance of the GP surrogate \citep{schobi2015polynomial}. 

GP-based surrogate models are widely used as predictor models in uncertainty quantification because they provide predictive means and variances that are analytically tractable and rely only on matrix products and inversions. 

Suppose we are given a training set of $n \geq 1$ input-output pairs $(\mathbf{x}, \mathbf{z})$ with $\mathbf{z} = z(\mathbf{x})$. For any $x_* \in \mathcal{X}$, one can obtain a predictive distribution at $x_*$ by conditioning the predictive distribution $z(\mathbf{x}, x_*)$ by the training data $(\mathbf{x}, \mathbf{z})$. The predictive distribution for the input $x_*$ after conditioning is given by:
\begin{equation*}
    z^{(n)}(x_*) \sim \mathcal{N}\left(m_n(x_*), \Cn(x_*) \right) 
\end{equation*}
\begin{equation*}
    m_n(x_*) = \mathbf{K}(x_*, \mathbf{x}) \mathbf{K}(\mathbf{x})^{-1} \mathbf{z} 
\end{equation*}
\begin{equation*}
    \Cn(x_*) =  \mathbf{K}(x_*, x_*) - \mathbf{K}(x_*, \mathbf{x}) \mathbf{K}(\mathbf{x})^{-1} \mathbf{K}(x_*, \mathbf{x})^T.
\end{equation*}

The hyperparameters of the GP model are the latent covariance kernels hyperparameters and the linear combination coefficients in the LMC model. For stochastic computer codes, we may add one hyperparameter per output dimension, which is known as the nugget noise level. The hyperparameters are tuned by maximization of the log-marginal likelihood $\log p(\mathbf{z} | \mathbf{x})$ with gradient-descent type methods \citep{byrd1995limited}. The log-marginal likelihood is given by:
\begin{equation}\label{eq:lml}
    \log p(\mathbf{z} | \mathbf{x}) = - \frac{nd}{2} \log(2 \pi) -\frac{1}{2} \log | \mathbf{K}(\mathbf{x}) | - \frac{1}{2} \mathbf{z}^T \mathbf{K}^{-1} \mathbf{z}
\end{equation}
where $| \mathbf{K}(\mathbf{x}) |$ is the determinant of the matrix $\mathbf{K}(\mathbf{x})$.

\subsection{Surrogate models in Bayesian inverse problems}
How are GP surrogates used in practice in Bayesian inverse problems? Starting from the consideration that the direct model $f$ is too costly, one can build a GP surrogate with some training data $(\mathbf{x}, \mathbf{z})$ obtained by calls to the function $f$. We wish to obtain a posterior distribution $p(\ \cdot \ | \mathbf{y})$ given some noisy observations $\mathbf{y}$ defined in \eqref{eq:observation_def}. The objective is then to include the predictive distribution of the GP model in the likelihood of the inverse problem.

The uncertainties in the inverse problem resolution can be split between the aleatoric uncertainty, linked to the noise of the observations, and the epistemic uncertainty which is the uncertainty of the surrogate model itself \citep{lartaud2023multi}. The two sources of uncertainty are independent from one another and one can express the observations as noisy predictions of the GP predictive mean, with two independent sources of errors:
\begin{equation*}
    \mathbf{y} = \mathbf{m}_n(x) + \boldsymbol{\eta}(x) + \boldsymbol{\varepsilon}
\end{equation*}
where $\mathbf{m}_n(x) = (m_n(x), ..., m_n(x))^T \in \R^{Nd}$, $\boldsymbol{\eta}(x) \sim \mathcal{N}\left(\mathbf{0}, \Cn(x) \otimes \mathcal{U}_N \right)$ and $\boldsymbol{\varepsilon} \sim \mathcal{N}\left(\mathbf{0}, \Cobs \otimes \mathcal{I}_N \right)$ are respectively the epistemic and aleatoric uncertainty terms. $\mathcal{I}_N$ is the identity matrix of size $N \times N$ and $\mathcal{U}_N$ is the matrix of ones of size $N \times N$. The notation $\otimes$ refers to the Kronecker product between matrices. 

This formulation expresses the idea that the observations are independent with respect to the observation noise, but are all linked together by the surrogate model error. The inverse problem can thus be solved by computing a posterior distribution with a new likelihood which accounts for both sources of uncertainties:
\begin{equation}\label{eq:total_density}
    p_n(x | \mathbf{y}) \propto L_n(\mathbf{y} | x) p(x) = \frac{1}{\sqrt{(2 \pi )^{Nd} | \Ctot(x) |}} \exp\left[- \frac{1}{2} \| \mathbf{y}_{\mathrm{flat}} - \mathbf{m}_n(x) \|^2_{\Ctot(x)} \right] p(x)
\end{equation}
with $\Ctot(x) = \Cn(x) \otimes \mathcal{U}_N + \Cobs \otimes \mathcal{I}_N \in \R^{Nd \times Nd}$, $\mathbf{y}_{\mathrm{flat}} \in \R^{Nd}$ is the flattened vector of observations and where $L_n(\mathbf{y} | x)$ is the likelihood in this expression.

This expression can be further simplified. We can show that the posterior density is proportional to:
\begin{equation*}
    p_n(x | \mathbf{y}) \propto L_n(\mathbf{y} | x) p(x)\propto |\Ceff(x)|^{-1/2} \exp \left(-\frac{1}{2} \|\overline{y} - m_n(x) \|^2_{\Ceff(x)}\right) p(x)
\end{equation*}
with $\Ceff(x) = \Cn(x) + \frac{1}{N}\Cobs \in S_d(\R)^+$ and $\overline{y} = \frac{1}{N} \sum\limits_{k=1}^N y^{(k)} \in \R^d$. This simplified expression is used in numerical applications, but the theoretical work is based on \eqref{eq:total_density}. This posterior distribution quantifies both the uncertainties from the observations and the uncertainties of the surrogate model. It can be obtained by standard MCMC sampling. 

\subsection{Stepwise uncertainty reduction}

The Stepwise Uncertainty Reduction (SUR) paradigm was introduced in \cite{bect2012sequential} and \cite{villemonteix2009informational}. It serves as a general-purpose theoretical framework for the design of computer experiments. SUR strategies seek to iteratively select new design points that minimize a given metric of uncertainty. The main advantage of SUR is to provide a theoretical foundation for many sequential design strategies. For instance, the SUR paradigm is applicable in Bayesian optimization with the efficient global optimization algorithm (EGO) \citep{jones1998efficient, snoek2012practical} or in reliability analysis \citep{chevalier2014fast, bect2017bayesian}. Before diving into our contribution, we provide a brief introduction to the SUR framework in this section.

Let $\mathbb{F}$ be a functional space, and $\mathbb{M}$ the set of Gaussian measures on $\mathbb{F}$. In our problem, the functional space considered is the space of continuous $\R^d$-valued functions on a compact space $\mathcal{X} \subset \R^p$ such that $\mathbb{F} = \mathcal{C}(\mathcal{X})$. In everything that follows, we will consider Gaussian processes with continuous sample paths, which can be understood as random elements of $\mathbb{F}$ \citep{vakhania1987probability, bogachev1998gaussian}. In particular, \citep{van2008reproducing} showed the equivalence between Gaussian measures and GP with continuous sample paths. 

Throughout this section, we consider a Gaussian measure $\nu \in \mathbb{M}$ and a GP with distribution $\nu$. Its distribution at $x \in \mathcal{X}$ is written as $z(x) \sim \mathcal{N}\left(m_\nu(x), \Cnu(x) \right)$.
\begin{definition}
    A sequential design is a collection $(x_n)_{n \geq 1} \in \mathcal{X}^n$ such that for all $n \geq 1$, $x_{n+1}$ is $\mathcal{F}_{n}$-measurable, with $\mathcal{F}_{n}$ the $\sigma$-algebra generated by the collection $(x_1, z_1, ..., x_n, z_n)$ with $z_i = z(x_i)$.
\end{definition}

We know that conditioning $z$ with a finite number $n$ of input-output pairs $(\mathbf{x}, \mathbf{z})$ yields a GP (with continuous sample paths). Consider a sequential design $(x_n)_{n \geq 1}$. We denote by $\nu_n$ the probability distribution of the conditioned GP. We know that $\nu_n$ is a Gaussian measure. The $z_n$ are given by calls to the direct model. Prior to these calls, they are unknown. In the SUR framework, the $z_n$ are modeled as random variables by $z_n = z(x_n)$. Because the $z_n$ are random variables, $\nu_n$ can be understood as a random element of $\mathbb{M}$. 

More precisely, \citep{bect2016supermartingale} showed that there exists a measurable mapping 
\begin{align*}
(\mathcal{X} \times \R^d)^n \times \mathbb{M}  &\longrightarrow \mathbb{M}\\
(x_1, z_1, ..., x_n, z_n, \nu) &\longmapsto \nu | (x_1, z_1, ..., x_n, z_n)
\end{align*}
such that for any Gaussian process $z$ with probability measure $\nu \in \mathbb{M}$ and any sequential design $(x_n)_{n \geq 1}$, the measure $\nu | (x_1, z_1, ..., x_n, z_n)$ is the conditional distribution of $z$ given $\mathcal{F}_n$. With our notation, $\nu_n = \nu | (x_1, z_1, ..., x_n, z_n)$.
 
Let us consider a measurable functional $\mathcal{H} \colon \mathbb{M} \xrightarrow[]{} \R_+$ which will serve as the metric of uncertainty that we would like to minimize. The goal of SUR strategies is to find a sequential design $(x_n)_{n \geq 1}$ which guarantees the almost sure convergence of the metric of uncertainty towards $0$:
\begin{equation*}
   \mathcal{H}(\nu_n) \xrightarrow[n \to +\infty]{a.s.} 0.
\end{equation*}

\begin{definition}
    Consider a functional $\mathcal{H} \colon \mathbb{M} \xrightarrow[]{} \R_+$. The functional $\mathcal{H}$ is said to have the supermartingale property if and only if, for any Gaussian process $z$ with continuous sample paths, the sequence $\mathcal{H}(\nu_n)$ is a $\mathcal{F}_n$-supermartingale or in other words for any $x \in \mathcal{X}$:
    \begin{equation*}
        \mathbb{E}_{n, x} \left[\mathcal{H}(\nu_{n+1}) \right] \leq\mathcal{H}(\nu_n)
    \end{equation*}
    where $ \mathbb{E}_{n, x}$ denotes the conditional expectation given $\mathcal{F}_n$ and $x_{n+1} = x$. We recall that $\nu_n$ is the probability distribution of the GP obtained by conditioning $z$ by $(\mathbf{x}, \mathbf{z})$.
\end{definition}

\begin{definition}
    If $\mathcal{H}$ has the supermartingale property, for $x \in \mathcal{X}$ let us introduce the functional $\mathcal{J}_x\colon \mathbb{M} \xrightarrow[]{} \R_+$ which is the expectation of the metric of uncertainty when conditioning the GP with respect to a new data point $(x, z(x))$ and given by:
    \begin{equation*}
        \mathcal{J}_x(\nu) = \mathbb{E}_{z(x)} \left[\mathcal{H}\left(\nu | (x, z(x)) \right) \right]
    \end{equation*}
    for $\nu \in \mathbb{M}$ and where $z(x) \sim \mathcal{N}\left(m_{\nu}(x), \Cnu(x) \right)$ and $\mathbb{E}_{z(x)}$ refers to the expectation with respect to the random output $z(x)$.
    We also introduce the functional $\mathcal{G}\colon \mathbb{M} \xrightarrow[]{} \R_+$ defined for $\nu \in \mathbb{M}$ by:
    \begin{equation}\label{eq:G_def}
        \mathcal{G}(\nu) = \sup_{x \in \mathcal{X}} \left(\mathcal{H}(\nu) - \mathcal{J}_x(\nu) \right).
    \end{equation}
    The sets of zeros of $\mathcal{H}$ and $\mathcal{G}$ are denoted respectively by $\mathbb{Z}_{\mathcal{H}}$ and $\mathbb{Z}_{\mathcal{G}}$. The inclusion $\mathbb{Z}_{\mathcal{H}} \subset \mathbb{Z}_{\mathcal{G}}$ is always true because $0 \leq \mathcal{G}(\nu) \leq \mathcal{H}(\nu)$. 
\end{definition}

SUR strategies are obtained by minimizing the expectation of the functional $\mathcal{H}$ conditioned by the knowledge of the yet unknown new design point. 
\begin{definition}
    A SUR sequential design for the functional $\mathcal{H}$ is defined as a sequential design such that for all $n \geq n_0$ with $n_0$ a given integer:
    \begin{equation*}
        x_{n+1} \in \argmin_{x \in \mathcal{X}} \left\{\mathbb{E}_{n, x} \left[\mathcal{H}(\nu_{n+1}) \right] \right\} = \argmin_{x \in \mathcal{X}} \mathcal{J}_x(\nu_n).
    \end{equation*}
    For the rest of this paper, let us introduce the notations $H_n = \mathcal{H}(\nu_n)$ and
    \begin{equation}\label{eq:def_Jn}
        J_n(x) = \mathbb{E}_{z^{(n)}(x)} \left[\mathcal{H}\left(\nu_n | (x, z) \right) \right] = \mathbb{E}_{n, x} \left[\mathcal{H}(\nu_{n+1}) \right]
    \end{equation}
    where $z^{(n)}(x) \sim \mathcal{N}\left(m_n(x), \Cn(x) \right)$ which will be frequently used in the next sections.
\end{definition}

\begin{definition}
    For any sequence of Gaussian measures $(\nu_n)_{n \geq 0}$, we say that the sequence ($\nu_n)_{n \geq 0}$ converges towards the limit measure $\nu_{\infty} \in \mathbb{M}$ when $(m_{\nu_n})_{n \geq 0}$ converges uniformly in $\mathcal{X}$ towards $m_{\nu_\infty}$, and $(k_{\nu_n})_{n \geq 1}$ converges uniformly in $\mathcal{X} \times \mathcal{X}$ towards $k_{\nu_\infty}$. 
\end{definition}

One can show that for any sequential design $(x_n)_{n \geq 1}$ and for any Gaussian process $z$ with continuous sample paths, $\nu_n = \nu | (x_1, z_1, ..., x_n, z_n)$ converges towards a limit measure $\nu_\infty \in \mathbb{M}$ \citep{bect2016supermartingale}. 
The following convergence theorem is presented in the aforementioned article and is used in this work to prove the convergence of our metric.
\begin{theorem}\label{conv_theorem}
    Let $\mathcal{H}$ be a non-negative uncertainty functional on $\mathbb{M}$ with the supermartingale property, and $\mathcal{G}$ be the functional defined in Equation \eqref{eq:G_def}. Consider $(x_n)_{n \geq 1}$ a SUR sequential design for $\mathcal{H}$. If $\mathbb{Z}_{\mathcal{H}} = \mathbb{Z}_{\mathcal{G}}$, $\mathcal{H}(\nu_n) \xrightarrow[n \to +\infty]{a. s.} \mathcal{H}( \nu_\infty)$ and $\mathcal{G}(\nu_n) \xrightarrow[n \to +\infty]{a. s.} \mathcal{G}( \nu_\infty)$, then $H_n = \mathcal{H}(\nu_n)\xrightarrow[n \to +\infty]{a. s.} 0$.
\end{theorem}

\section{Sequential design for Bayesian inverse problems}\label{sec:4}
Let us get back to the main problem at stake here, which is the Bayesian resolution of an ill-posed inverse problem. In this work, the following question is asked. Given a limited computational budget, how can new design points be added to train the surrogate model? In this section, a first strategy based on a D-optimal design criterion is presented and then a second approach based on the SUR paradigm for a well-chosen metric of uncertainty is derived. This new strategy named IP-SUR (Inverse Problem SUR) is shown to be tractable for GP surrogates and is supported by a theoretical guarantee of almost sure convergence to zero for the metric of uncertainty.

\subsection{An introductory example}
In this section, we consider a scalar inverse problem to highlight some limitations of standard design strategies. As usual, consider a GP surrogate we wish to improve with a few model calls. Several design strategies can be carried out such as D-optimal designs or I-optimal designs \citep{antognini2010exact}. However, these methods do not account for the posterior distribution and may choose points that lie far from the posterior distribution. Although they may improve the surrogate, such design points are not valuable in this context. If they are in regions of low posterior densities, the improvements to the surrogate model will not be reflected in the MCMC sampling.

In this example, the direct model is a parabolic function $f(x) = \frac{x^2}{2}$. We act as if this direct model were a complex computer code and replace it with a surrogate model. Given one observation of the model $y = f(x_{\mathrm{th}}) + \varepsilon$, with $x_{\mathrm{th}} = 3$ and $\varepsilon \sim \mathcal{N}\left(0, 0.5^2 \right)$ we evaluate the posterior distribution $p(x | y)$ using the surrogate model. Then, we add two new design points with various design strategies. The newly added design points are plotted in Figure \ref{fig:example_1D}, along with the posterior density obtained by MCMC and the GP prediction. Between the first and second iteration, the training point selected by the IP-SUR strategy is added to the model. 
\fig{0.99}{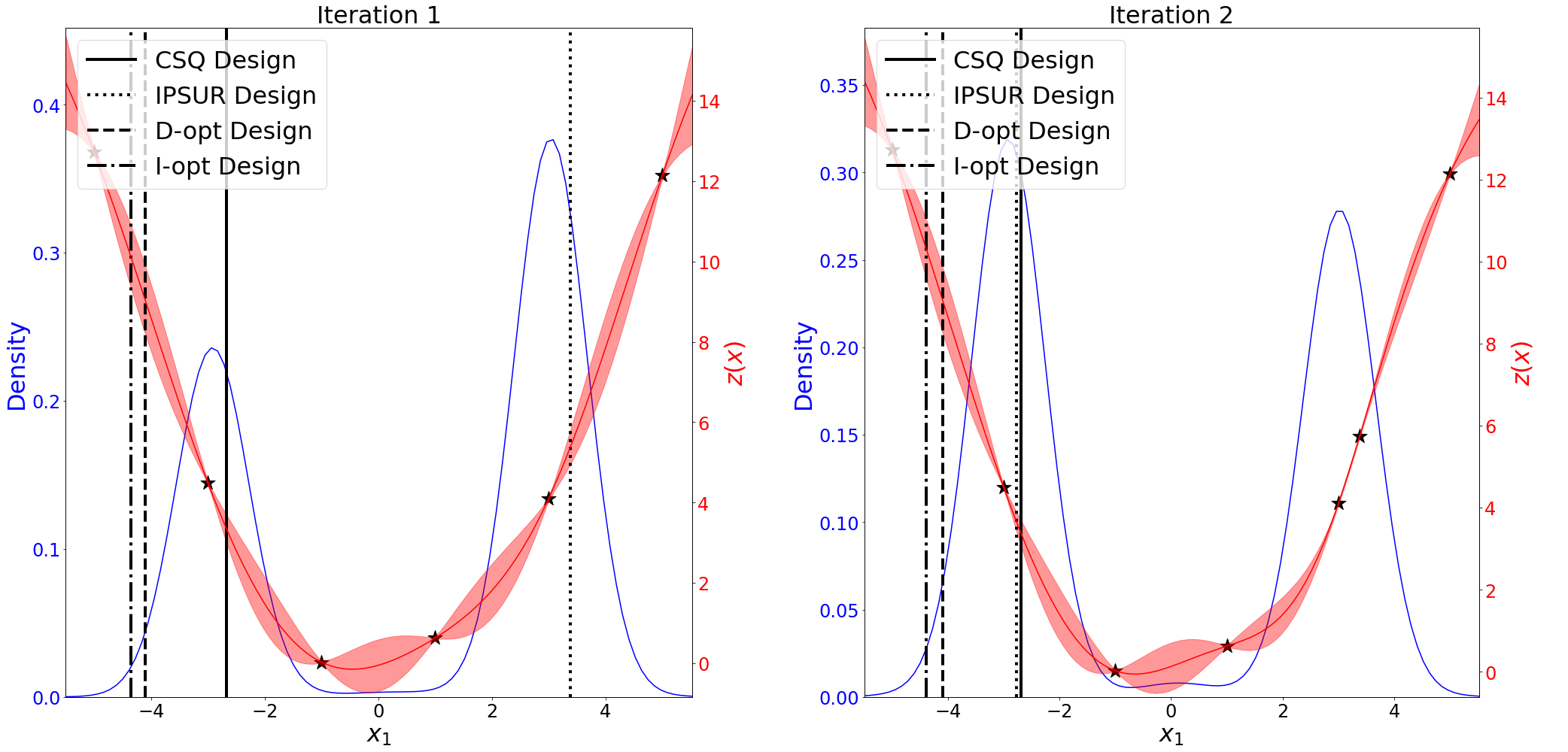}{Example of two iterations of various design strategies. The GP mean and 95 \% credible interval are plotted in red. The posterior density obtained by MCMC is plotted in blue. The vertical black lines represent the four points determined by the four different strategies.}{fig:example_1D}

The strategies investigated in this illustrative example are D-optimal, I-optimal and the two strategies proposed in this paper, named CSQ and IP-SUR. They are presented respectively in Section \ref{seq:csq_design} and Section \ref{sec:ipsur}.

In this illustrative example, the two proposed strategies select points for which both the posterior density and the model uncertainty are large. On the other hand, D-optimal and I-optimal strategies do not account for the posterior density and thus select points that may not improve the inverse problem resolution, because these points are selected in regions of low posterior density.

\subsection{Constraint set query}\label{seq:csq_design}
Our first idea is to adapt D-optimal design strategies to Bayesian inverse problems. To improve the surrogate in regions of high posterior density, the next training point $x_{n+1} \in \mathcal{X}$ can be chosen to minimize the determinant of the predictive covariance matrix on a well-chosen subset $\mathcal{B} \subset \mathcal{X}$:
\begin{equation}\label{eq:csq_strategy}
    x_{n+1} \in \mathrm{argmax}_{x \in \mathcal{B}} \ |\Cn(x)|.
\end{equation}
How can we choose $\mathcal{B}$? One could try to force the newly chosen design points to be near the maximum a posteriori (MAP) $\xmap$ defined by:
\begin{equation*}
    \xmap \in \argmax_{x \in \mathcal{X}} \ p_n(x | \mathbf{y} ).
\end{equation*}
Then for any $h \in \mathbb{R}_+$ one can define the subset $\mathcal{B}_n(h)$ such as:
\begin{equation}\label{eq:bounding_set}
    \mathcal{B}_n(h) = \left \{ x \in \mathcal{X} | \log p_n\left(\xmap | \mathbf{y} \right) - \log p_n(x| \mathbf{y}) \leq h  \right\}.
\end{equation}
To get an intuition of the influence of $h$, let us look at the acceptance probability $\alpha_n(x_1, x_2)$ for a jump from point $x_1$ to $x_2$ in the context of Metropolis-Hastings sampling with a symmetric proposal distribution:
\begin{equation*}
    \alpha_n(x_1, x_2) = \min \left\{ 1 , \frac{p_n(x_2 | \mathbf{y})}{p_n(x_1 | \mathbf{y})} \right\} = \min \left\{ 1 , \frac{L_n(\mathbf{y} | x_2) p(x_2)}{L_n(\mathbf{y} | x_1) p(x_1)} \right\}.
\end{equation*}
Then $\mathcal{B}_n(h) = \left \{ x \in \mathcal{X} | \log \alpha_n(\xmap, x) \geq -h  \right\} = \left \{ x \in \mathcal{X} | \alpha_n(\xmap, x) \geq e^{-h}  \right\}$. Thus, $e^{-h}$ for $h \in \mathbb{R}_+$ represents the lowest possible acceptance probability of a Metropolis-Hastings jump from the MAP to $x \in \mathcal{B}_n(h)$. Depending on the choice of $h$, the new query point $x_{n+1}$ will be quite close to the MAP (if $h$ is close to $0$) or it can be far from the MAP (if $h$ is large). 

The CSQ method requires solving the optimization problem \eqref{eq:csq_strategy}. This is done with the dual annealing algorithm \citep{xiang1997generalized}. The choice of a log-decreasing temperature scheme $T_n = \frac{C}{\log n}$ in the annealing algorithm guarantees that the Markov chain reaches a neighborhood of the set of minima of the function $f$ to optimize \citep{locatelli2000simulated, bohachevsky1986generalized}.

This constraint set query (CSQ) method is a simple approach for active learning in the context of Bayesian inverse problems. It can be understood as a D-optimal sequential design strategy restricted to a subset of the domain. 

\subsection{Inverse Problem SUR (IP-SUR) strategy}\label{sec:ipsur}
Now, let us use the SUR framework to derive a second sequential design well-suited for Bayesian inverse problems. To any Gaussian measure $\nu$ on $\mathbb{F} = \mathcal{C}(\mathcal{X})$, we associate an equivalent GP $z$. We denote by $p_\nu(x_* | \mathbf{y})$ the posterior density obtained from \eqref{eq:total_density} with the GP surrogate $z$. This density is proportional to the product of a likelihood $L_\nu(\mathbf
y | x_*)$ and a prior density $p(x_*)$. Given $n$ training instances $(\mathbf{x}, \mathbf{z})$, the posterior density obtained from the conditioned GP is denoted by $p_n(x | \mathbf{y})$ as usual, and the corresponding likelihood is $L_n(\mathbf{y}|x_*)$. 

We introduce the functional $\mathcal{H}\colon \mathbb{M} \xrightarrow[]{} \R_+$ defined for any measure $\nu \in \mathbb{M}$ by:
\begin{equation}\label{eq:def_H}
    \mathcal{H}(\nu) = \mathbb{E}_{p_{\nu}} \left[ \ | \Cnu(\cdot) \right| \ ] = \int_{\mathcal{X}} |\Cnu(x_*)| p_{\nu}(x_* |\mathbf{y}) dx_*.
\end{equation}
This functional is a weighted integrated Mean Squared Prediction Error (weighted IMSPE). The weight is the posterior density, which focuses the attention on the region of interest for the inverse problem. 

In practice, we would like to write a SUR criterion for this functional $\mathcal{H}$. However, using $\mathcal{H}$ leads to two caveats. First, the SUR criterion is not tractable. The reason lies in the normalization constant $Z_n$ of the posterior distribution defined by:
\begin{equation}\label{eq:def_Cn}
	Z_n = \int_\mathcal{X} p(x_*) L_n(\mathbf{y} | x_*) dx_*.
\end{equation}
At iteration $n+1$, this normalization constant $Z_{n+1}$ depends on $z_{n+1}$ and complicates the computation of $J_n(x)$ defined in \eqref{eq:def_Jn}.

Additionally, for similar reasons, we cannot prove that $\mathcal{H}$ has the supermartingale property, which means that we cannot use Theorem \ref{conv_theorem} to show the convergence of $H_n$.

To overcome these difficulties, we introduce the functional $\mathcal{D}$ defined for any Gaussian measure $\nu$ by:
\begin{equation}\label{eq:def_D}
\mathcal{D}(\nu) = \int_\mathcal{X} |\Cnu(x_*)| p(x_*) L_\nu(\mathbf{y} | x_*) dx_*.
\end{equation}
One can notice that $D_n = \mathcal{D}(\nu_n) = \frac{H_n}{Z_n}$. We modify our objective and now we seek to write a SUR criterion for the functional $\mathcal{D}$ and show that under this design, we have the convergence of $H_n$ to zero almost surely.

\subsubsection{Evaluation of the metric}
Consider a GP surrogate model $z$, conditioned on $n$ pairs input-output. Let $x$ be a new design point and $z_{n+1} = z^{(n)}(x)$ the corresponding random output. At stage $n+1$, we condition the GP with the new datapoint $(x, z_{n+1})$ which updates the predictive distribution. For $x_* \in \mathcal{X}$, the updated mean and covariance functions of the newly conditioned GP are:
\begin{equation*}
    m_{n+1}(x_* | x, z_{n+1}) = m_n(x_*) + \Cn(x_*, x) \Cn(x)^{-1} (z_{n+1} - m_n(x))
\end{equation*}
\begin{equation}\label{eq:cov_update}
    \mathbf{C}_{n+1}(x_* | x) = \Cn(x_*) - \Cn(x_*, x) \Cn(x)^{-1} \Cn(x_*, x)^T.
\end{equation}
 We are interested in the variance integrated over the posterior distribution. The quantity of interest $D_n$ is given by: 
 \begin{equation*}
     D_n = \mathcal{D}(\nu_n) = \int_{\mathcal{X}} | \Cn(x_*)| L_n( \mathbf{y} | x_* ) p(x_*) dx_*.
 \end{equation*}
 When adding a new data point $(x, z_{n+1})$, the quantity of interest becomes:
 \begin{equation*}
     D_{n+1}(x, z_{n+1}) = \mathcal{D}(\nu_n | (x, z_{n+1})) = \int_{\mathcal{X}} | \mathbf{C}_{n+1}(x_*|x) | L_{n+1}(\mathbf{y} | x_*, x, z_{n+1}) p(x_*)  dx_*
 \end{equation*}
 where $L_{n+1}(\mathbf{y}|x_*, x, z_{n+1})$ is the updated likelihood obtained from \eqref{eq:total_density} and given by:
\begin{equation*}
    L_{n+1}(\mathbf{y} | x_*, x, z_{n+1}) = ((2 \pi)^{Nd} | \Ctotup(x_* | x) | )^{-1 / 2} \exp \left[- \frac{1}{2}\left\| \mathbf{y - m}_{n+1}(x_*| x, z_{n+1}) \right\|^2_{\Ctotup(x_* | x)} \right]
\end{equation*} 
\begin{equation*}
    \Ctotup(x_*| x) = \Ctot(x_*) - \boldsymbol\lambda_n(x, x_*) \otimes \mathcal{U}_N
\end{equation*}
with $\boldsymbol\lambda_n(x, x_*) = \Cn(x_*, x) \Cn(x)^{-1} \Cn(x_*, x)^T$.

Our objective is to minimize the metric of uncertainty $D_n$. For that purpose, the stepwise uncertainty reduction paradigm is adapted for this specific problem. For any $x \in \mathcal{X}$, we introduce $F_n(x)$ defined by:
\begin{equation}
    F_n(x) = \mathbb{E}_{z_{n+1}} \left[D_{n+1}(x, z_{n+1}) \right]
\end{equation}
where $z_{n+1} \sim \mathcal{N}\left(m_n(x), \Cn(x) \right)$ and $\mathbb{E}_{z_{n+1}}$ is the expectation w.r.t. $z_{n+1}$. This quantity is the equivalent, for the functional $\mathcal
{D}$, of $J_n(x)$ defined in \eqref{eq:def_Jn}.

The SUR strategy associated to the functional $\mathcal{D}$ is thus given by:
\begin{equation}\label{eq:sur_strategy}
    x \in \argmin_{x \in \mathcal{X}} F_n(x).
\end{equation}

This criterion defines the IP-SUR strategy. We can claim the IP-SUR strategy is interesting if we can solve the following two problems.
First, we need to make sure the quantity $F_n(x)$ can be evaluated in a closed form (up to a multiplicative constant) to solve \eqref{eq:sur_strategy}. Otherwise, a suitable approximation must be used as is done in \cite{zhang2019sequential} for an expected improvement criterion for example. Second, we need to show the convergence of the functional of interest $H_n = \mathcal{H}(\nu_n)$.

The first focus is on $F_{n}(x) = \mathbb{E}_{z_{n+1}} \left[ D_{n+1}(x, z_{n+1}) \right] $. Since in our work, we have access to an ergodic Markov chain for the posterior distribution, we want to write $F_{n}(x)$ as an expectation with respect to this posterior and use the ergodicity of the chain to evaluate the expectation.

\begin{proposition}\label{prop:sur_criteria}
    The quantity $F_n(x)$ is given for $x \in \mathcal{X}$ by:
\begin{equation}\label{eq:hat_jn}
    F_n(x) = \int_\mathcal{X} |\mathbf{C}_{n+1}(x_* | x) |L_n(\mathbf{y} |x_*) p(x_*) dx_*
\end{equation}
\end{proposition}
\begin{proof}
    The proof is given in Appendix \ref{app:proof_sur_criteria}
\end{proof}
The SUR strategy is unchanged when multiplying by a constant (i.e. a quantity independent of $x$). The criterion \eqref{eq:sur_strategy} is thus equivalent to the criterion:
\begin{equation*}
    x \in \argmin_{x \in \mathcal{X}} \frac{F_n(x)}{Z_n} .
\end{equation*}
With \eqref{eq:hat_jn}, we show that this new criterion boils down to an optimization problem where the target function $\frac{F_n(x)}{Z_n}$ can be evaluated by:
\begin{equation*}
    \frac{F_n(x)}{Z_n} \simeq \frac{1}{L} \sum_{l=1}^L |\mathbf{C}_{n+1}(X_l^{(n)} | x) |
\end{equation*}
where $(X_l^{(n)})_{1 \leq l \leq L}$ is an ergodic Markov chain whose invariant distribution is the posterior $p_n(\ \cdot \ | \mathbf{y})$. This Markov chain can be obtained by standard MCMC sampling methods. Once again, the dual annealing algorithm is used to solve the optimization problems.

Therefore, we showed that the SUR criterion is tractable. We then prove the almost sure convergence of $\mathcal{H}(\nu_n)$ toward zero, when following the IP-SUR strategy.

\subsubsection{Convergence of the IP-SUR sequential design}
We begin by proving the supermartingale property for the auxiliary functional $\mathcal{D}$. This lemma is then used to derive the main convergence theorem.
\begin{lemma}\label{prop:smp_for_D}
    The functional $\mathcal{D}\colon \mathbb{M} \xrightarrow[]{} \R_+$ defined for any Gaussian measure $\nu \in \mathbb{M}$ by:
    \begin{equation*}
        \mathcal{D}(\nu) =  \int_{\mathcal{X}} |\mathbf{C}_{\nu}(x_*)| L_{\nu}(\mathbf{y} | x_*) p(x_*) d x_*
    \end{equation*}
    has the supermartingale property. In other words, for any sequential design $(x_n)_{n \geq 1}$, there exists $n_0 \geq 1$ such that for all $n \geq n_0$, and for all $x \in \mathcal{X}$
    \begin{equation*}
        \mathbb{E}_{n, x} \left[\mathcal{D}(\nu_n | x, z(x)) \right] \leq \mathcal{D}(\nu_n).
    \end{equation*}
\end{lemma}
\begin{proof}
    The proof of this proposition is given in Appendix \ref{app:proof_smp}.
\end{proof} 
Note that $\mathcal{D}(\nu) = Z_\nu \times \mathcal{H}(\nu)$ with $Z_\nu = \int_\mathcal{X} L_\nu(\mathbf{y} | x_*) p(x_*) dx_*$ being the normalization constant of the posterior distribution.
\begin{lemma}\label{lemma}
    If $Z_n = \int_{\mathcal{X}} L_n(\mathbf{y}
    | x_*) p(x_*) dx_*$ for $n \geq 1$, then the sequence $(Z_n)_{n \geq 1}$ converges almost surely and its limit is positive and given by:
    \begin{equation*}
        Z_{\infty} = \int_{\mathcal{X}} L_\infty(\mathbf{y} |x_*) p(x_*) dx_* 
    \end{equation*}
    where $L_\infty(\mathbf{y} |x_*)$ is defined for $x_* \in\mathcal{X}$ by:
    \begin{equation*}
        L_\infty(\mathbf{y} | x_*) = ((2 \pi)^{Nd} | \boldsymbol\Sigma_\infty(x_*) | )^{-1 / 2} \exp \left[- \frac{1}{2}\left\| \mathbf{y - m_\infty}(x_*)\right\|^2_{\boldsymbol\Sigma_\infty(x_*)} \right]
    \end{equation*}
    with $m_\infty$ and $\mathbf{C}_\infty$ being the respective limits of the GP mean function $(m_n)_{n \geq 1}$ and covariance function $(\Cn)_{n \geq 1}$ and $\boldsymbol\Sigma_{\infty}(x_*) = \mathbf{C}_\infty(x_*) \otimes \mathcal{U}_N + \Cobs \otimes \mathcal{I}_N$ for $x_* \in \mathcal{X}$.
\end{lemma} 
\begin{proof}
    The proof is given in Appendix \ref{app:proof_lemma}
\end{proof}
Then, let us investigate the convergence of the functional $\mathcal{H}$. Theorem \ref{conv_theorem} is our main tool to prove the convergence of $\mathcal{H}$. 
\begin{theorem}\label{th:conv_H}
    Consider the functional $\mathcal{H}$ defined in \eqref{eq:def_H} and a SUR sequential design $(x_n)_{n \geq 1}$ given by the strategy \eqref{eq:sur_strategy}, associated to the functional $\mathcal
    D$. Then, the metric of uncertainty $H_n$ converges almost surely to $0$:
    \begin{equation*}
        H_n \xlongrightarrow[n \to +\infty]{a.s.} 0.
    \end{equation*}
\end{theorem}
\begin{proof}
The proof is given in Appendix \ref{app:conv_proof}.
\end{proof}
\begin{remark}
    Theorem \ref{conv_theorem} can be extended to quasi-SUR sequential design, as stated by \cite{bect2016supermartingale}. $(x_n)_{n \geq 1}$ is a quasi-SUR sequential design if there exists a sequence $(\varepsilon_n)_{n \geq 1}$ of non-negative real numbers such that $\varepsilon_n \xrightarrow{n \to + \infty} 0$, and if there exists $n_0 \in \mathbb{N}$ such that $(x_n)_{n \geq 1}$ verifies $F_n(x_{n+1}) \leq \inf_{x \in \mathcal{X}} F_n(x) + \varepsilon_n$ for all $n \geq n_0$. Thus, the almost sure convergence proven in Theorem \ref{th:conv_H} is also valid for quasi-SUR designs.
    
    This remark is crucial for numerical applications since we are only guaranteed to reach a neighborhood of the global minimum in the optimization step for the IP-SUR strategy. The convergence theorem for the quasi-SUR designs is more flexible in that regard and ensures convergence even for numerical applications. 
\end{remark}

\section{Application}
\subsection{Set-up and test cases}
The IP-SUR strategy presented in this paper is applied to various test cases for different shapes of posterior distributions and is then compared to the CSQ approach described in Section \ref{seq:csq_design}, to D-optimal and I-optimal designs, and to a state-of-the-art sequential design based on Bayes risk minimization. These applications rely on multi-output Gaussian process surrogate models. In the first two test cases, we are considering simple test cases in two dimensions. In the last test case, the strategies are applied to a more difficult example taken from neutron noise analysis, a technique whose goal is to identify fissile material from measurements of temporal correlations between neutrons. 

In everything that follows, the direct model is always considered too costly to be called directly and is replaced by a GP surrogate. The observations of the direct model are noisy with a known noise covariance $\Cobs$. The multi-output Gaussian process surrogate model is based on the Linear Model of Coregionalization described in \citep{bonilla2007multi}. To build this surrogate model, a training dataset of $n_0=10$ input-output pairs $(\mathbf{x}, \mathbf{z})$ is accessible. The surrogate model obtained is denoted $z^{(0)}$. 

This initial surrogate model is our starting point from which a sequential design is built. At each iteration, a new design point is acquired with the tested strategies, and the GP surrogate is updated to form a new surrogate $z^{(n)}$. The posterior distribution obtained with the initial GP is denoted $p_0(x_* | \mathbf{y})$ and after $n \geq 1$ iterations, the new posterior distribution is $p_n(x_* | \mathbf{y})$. Similarly, $L_0(\mathbf{y} | x_*)$ and $L_n(\mathbf{y} | x_*)$ are the initial and updated likelihoods in the inverse problem. 

At each iteration, a new Markov chain $(X_l^{(n)})_{1 \leq l \leq L}$ is generated with $L = 2 \times 10^5$ with the Adaptive Metropolis algorithm \cite{haario2001adaptive}. The prior in the inverse problem is always uniform on the input domain. 

To compare the different approaches for sequential design, performance metrics are needed. The first obvious choice is the integrated variance (IVAR) as introduced in \eqref{eq:def_H}. The iterative values $H_n = \mathcal{H}(\nu_n)$ of the metric are evaluated for each iteration. 
The next metric considered is the differential entropy of the posterior distribution:
\begin{equation*}
    S_n = - \int_{\mathcal{X}} p_n( x_*|\mathbf{y}) \log p_n( x_*|\mathbf{y}) dx_*.
\end{equation*}
Because we have access to ergodic Markov chains $(X_l^{(n)})_{1 \leq l \leq L}$ and the prior is uniform, the differential entropy can be estimated by:
\begin{equation*}
    S_n \simeq \widehat{S_n} = - \frac{1}{L} \sum\limits_{l=1}^L \log p_n(X_l^{(n)} |\mathbf{y})
\end{equation*}
where the posterior densities are obtained from kernel density estimation using the MCMC samples. 

The other performance metric considered is the Kullback-Leibler \citep{kullback1951information} divergence $\mathrm{KL}\left[p_n(\ \cdot \ | \mathbf{y}) \| p_{\infty}(\ \cdot \ | \mathbf{y})\right]$ between a posterior $p_n(\ \cdot \ | \mathbf{y})$ and a reference posterior distribution $p_{\infty}(\ \cdot \ | \mathbf{y})$ obtained from a GP surrogate $z^{(\infty)}$ trained with infinitely many data points. For the numerical applications, we do not have access to $p_\infty$. We compare $p_n(\ \cdot \ | \mathbf{y})$ with the posterior obtained with a well-trained GP $p_{1000}(\ \cdot \ | \mathbf{y})$ instead. We thus define:
\begin{align*}
    \kappa_n &= \mathrm{KL}\left[p_n(\ \cdot \ | \mathbf{y}) \| p_{1000}(\ \cdot \ | \mathbf{y}) \right] = \int_{\mathcal{X}} p_n(x_* | \mathbf{y}) \log \left( \frac{p_n(x_* |\mathbf{y})}{p_{1000}(x_* |\mathbf{y})} \right) dx_*.
\end{align*}
To estimate the KL, the ergodicity of the Markov chain is used once again, in combination with the kernel density estimates for the posterior densities:
\begin{align*}
    \kappa_n \simeq \widehat{\kappa_n} = \frac{1}{L } \sum\limits_{l=1}^L \log \left(\frac{p_n(X_l^{(n)} |\mathbf{y})}{p_{1000}(X_l^{(n)} |\mathbf{y})} \right).
\end{align*}
To avoid introducing a bias, the samples $X_l^{(n)}$ used to evaluate the performance metrics should be independent from the ones used in the design strategies. However, using the same samples did not lead to any change in the final results.

\subsubsection{Banana posterior distribution}
In this first test case, the target posterior distribution has a banana shape as displayed in Figure \ref{fig:banana}.
\fig{0.6}{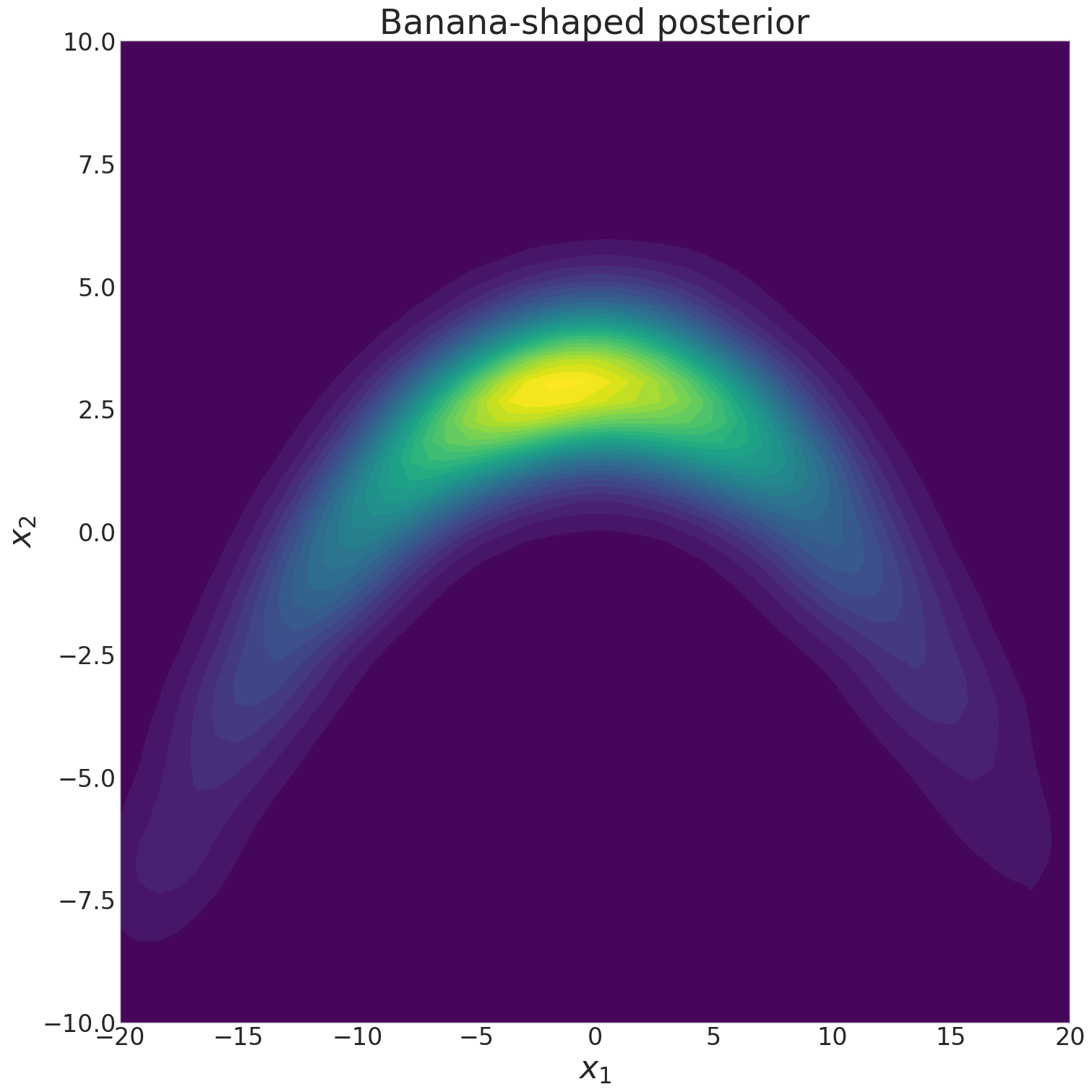}{Banana-shaped target posterior distribution}{fig:banana}
This posterior distribution is similar to the one introduced in \cite{surer2024sequential} and is described by the following analytical direct model:
\begin{align*}
    f_{\mathrm{b}} \colon \mathcal{X}_{\mathrm{b}} &\longrightarrow \R^2 \\
    (x_1, x_2) &\longmapsto (x_1, x_2 + 0.03 x_1^2)
\end{align*}
where $\mathcal{X}_{\mathrm{b}}= [-20, 20] \times [-10, 10] \subset \R^2$.
For a single observation $y = (y_1, y_2)$, the posterior has the density:
\begin{equation*}
p_{\mathrm{b}}(x | y) \propto \exp \left(- \frac{1}{2} \frac{(x_1 - y_1)^2}{100} - \frac{1}{2} \left(x_2 + 0.03 x_1^2 - y_2\right)^2 \right) .
\end{equation*}
The observations $\mathbf{y} = (y^{(k)})_{1 \leq k \leq N}$ are generated with $N = 5$ and such that for $1 \leq k \leq N$ we have $y^{(k)} \sim \mathcal{N}\left(\mu,  \Cobs \right)$ with $\mu = (0, 3)$ and:
\begin{equation*} 
     \Cobs = 
    \begin{pmatrix}
        100 & 0 \\
        0 & 1 \\
    \end{pmatrix}.
\end{equation*}

\subsubsection{Bimodal posterior distribution}
In this second test case, the target posterior is bimodal as plotted in Figure \ref{fig:bimodal}.
The corresponding direct model is $f_m$ defined by:
\begin{align*}
    f_{\mathrm{m}} \colon \mathcal{X}_{\mathrm{m}} &\longrightarrow \R^2 \\
    (x_1, x_2) &\longmapsto (x_2 - x_1^2, x_2 - x_1)
\end{align*}
where $\mathcal{X}_{\mathrm{m}} = [-6, 6] \times [-4, 8] \subset \R^2$.
For a single observation $y = (y_1, y_2)$, the posterior has the density:
\begin{equation*}
p_{\mathrm{m}}(x | y) \propto \exp \left(- \frac{\sqrt{0.2}}{10} (x_2 - x_1^2 - y_1)^2 - \frac{\sqrt{0.75}}{10} \left(x_2 + 0.03 x_1^2 - y_2\right)^2 \right).
\end{equation*}
The observations $\mathbf{y} = (y^{(k)})_{1 \leq k \leq N}$ are generated with $N = 10$ and such that for $1 \leq k \leq N$ we have $y^{(k)} \sim \mathcal{N}\left(\mu,  \Cobs \right)$ with $\mu = (0, 2)$ and:
\begin{equation*} 
     \Cobs = 
    \begin{pmatrix}
        \frac{5}{\sqrt{0.2}} & 0 \\
        0 & \frac{5}{\sqrt{0.75}} \\
    \end{pmatrix}.
\end{equation*}
\fig{0.6}{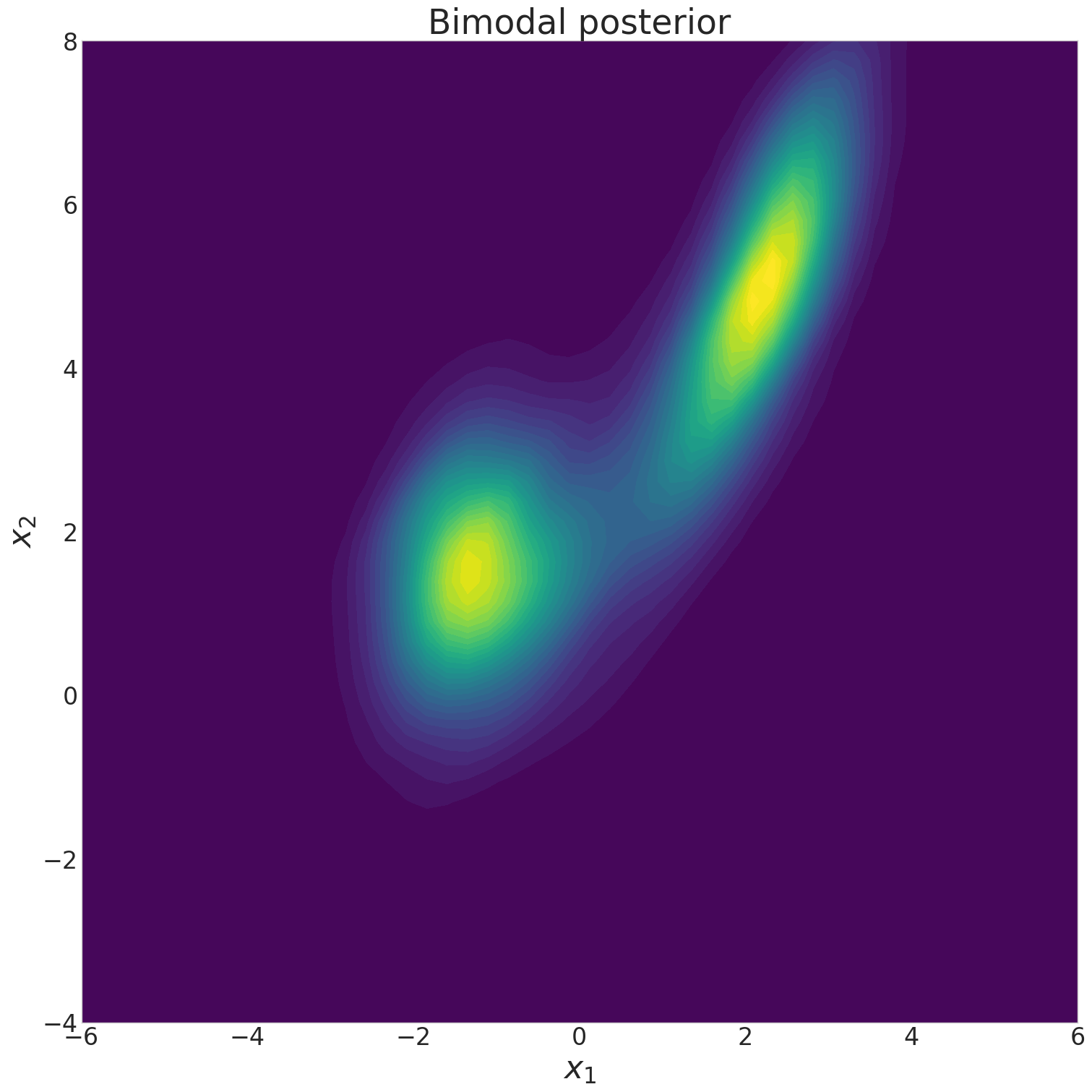}{Bimodal target posterior distribution}{fig:bimodal}

\subsubsection{A practical application to neutron noise analysis}
Now let us consider an example taken from a practical problem in neutron noise analysis \citep{pazsit2007neutron}. In this application, the direct model is the analytical approximation of a more complex direct model which involves a full 3D neutron transport model. 

Let us consider the following problem. Let $f_{\mathrm{p}} : \mathcal{X}_{\mathrm{p}} \mapsto \R^3$ be the analytical direct model defined in Appendix \ref{app:neutronic}, with $\mathcal{X}_{\mathrm{p}} \subset \R^4$. Let $\mathbf{y} = (y^{(k)})_{1 \leq k \leq N}$ be the $N = 20$ noisy observations of the direct model:
\begin{equation*}
    y^{(k)} = f_{\mathrm{p}}(x_{\mathrm{th}}) + \varepsilon^{(k)} \text{ with } \varepsilon^{(k)} \sim \mathcal{N} \left(\mathbf{0}, \Cobs \right)
\end{equation*}
where $x_{\mathrm{th}} \in \mathcal{X}_{\mathrm{p}}$ is the true value of the inputs. The domain of interest for the input parameters is:
\begin{equation*}
    \mathcal{X}_{\mathrm{p}} = [0.7, 0.9] \times [0.01, 0.10] \times [1 \times 10^5, 2 \times 10^5] \times [0.1, 0.9].
\end{equation*}
Some two-dimensional marginals of the posterior distribution $p_{1000}$ obtained with a well-trained GP are shown in Figure \ref{fig:posterior_neutronic}.
\fig{0.99}{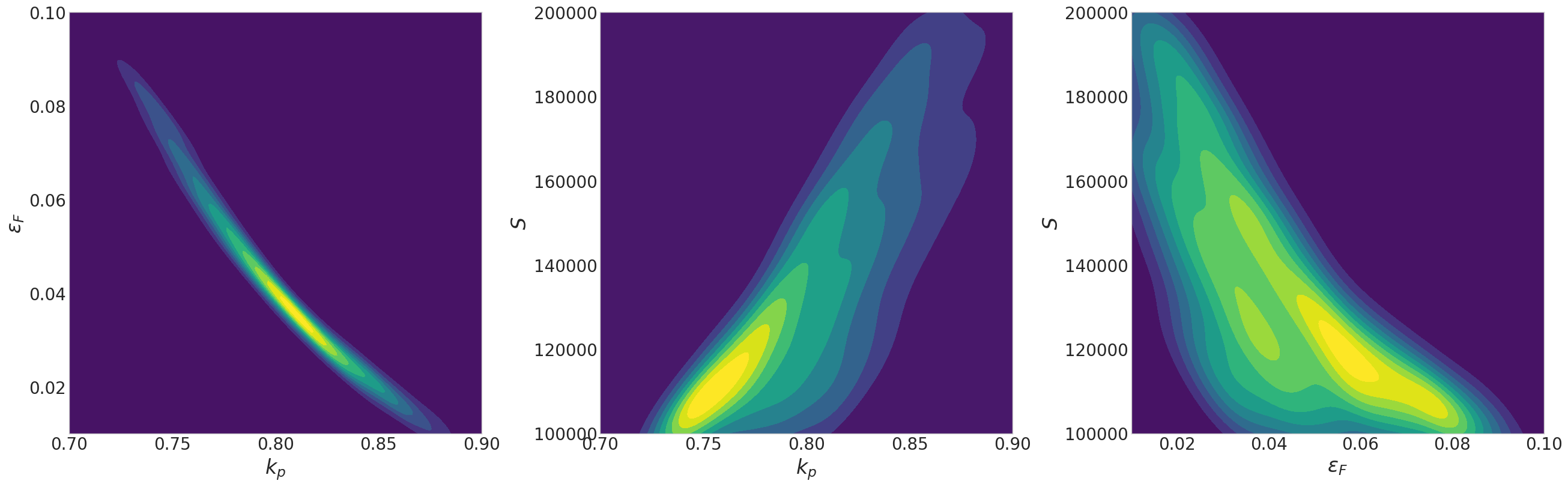}{Two-dimensional marginal distributions from the target posterior distribution in the neutron noise analysis case study with $x=(k_p, \varepsilon_F, S, x_s)$}{fig:posterior_neutronic}
This posterior distribution is more difficult to sample as it mostly lies on a one-dimensional manifold subspace in the four-dimensional parameter space. The decorrelation time $\tau$ is significantly larger with $\tau \simeq 200$ compared to $\tau \simeq 10$ for the two previous test cases. Thus, the number of MCMC samples is increased to $4 \times 10^5$ iterations per run for that specific test case.

\subsection{Results}
The IP-SUR sampling strategy and the CSQ method both require solving an optimization problem with a non-convex function. We recall the optimization problems are solved with the dual annealing algorithm as implemented in \textit{scipy}. The same is true for the other strategies.

The IP-SUR strategy and the CSQ method are applied iteratively $20$ times to produce new design points. For the CSQ strategy, the hyperparameter $h$ introduced in \eqref{eq:bounding_set} is set to $h=3$. The influence of $h$ is discussed afterward.

The average running times (in seconds) for the tested strategies are presented in Table \ref{tab:run_times}. They are obtained on a local computer equipped with an 11th Gen Intel(R) Core(TM) i7-11800H @ 2.30GHz processor. 
\begin{table}[h!]
\caption{Average running times (in seconds) for the various strategies, on each of the test cases.}
\vspace{0.2cm}
\label{tab:run_times}
    \centering
    \begin{tabular}{|c|c c c c c|}
    \hline
      & CSQ & IP-SUR & D-optimal & I-optimal & Sinsbeck \textit{et al.} \\
    \hline
    Banana case    &  $16$ & $235$  & $3$  & $136$ & $180$   \\
    Bimodal case   &  $14$ & $265$  & $2$  & $92$  & $156$   \\
    Neutronic case &  $18$ & $1676$ & $16$ & $440$ & $1235$   \\
    \hline
    \end{tabular}
\end{table}
\fig{0.99}{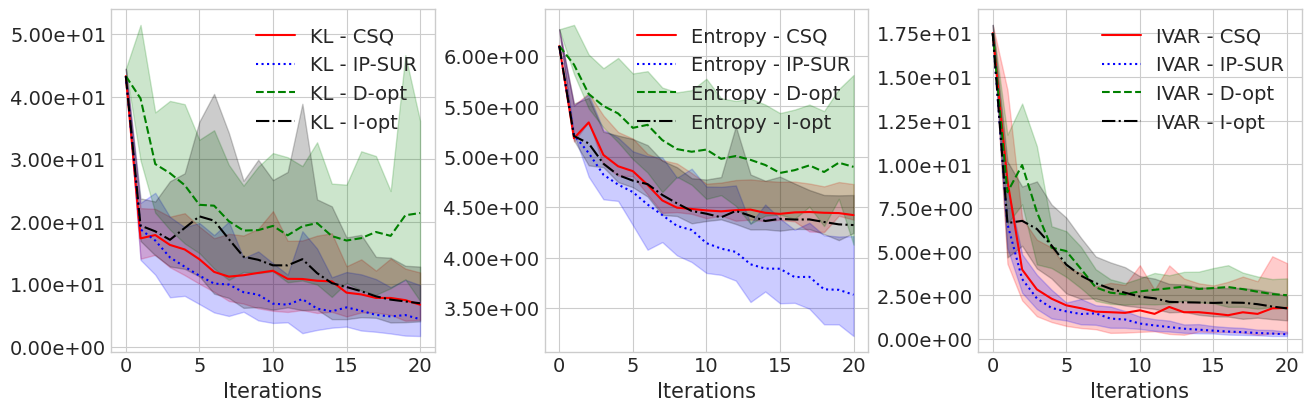}{Mean evolution and $95 \%$ confidence intervals of the performance metrics over $20$ iterations of the CSQ, IP-SUR, D-,optimal and I-optimal design strategies, replicated for $100$ experiments - Banana test case}{fig:banana_metrics}
The evolutions of the performance metrics are plotted in Figures \ref{fig:banana_metrics}, \ref{fig:bimodal_metrics}, and \ref{fig:neutronic_metrics} for all tested strategies. The empirical $95$ \% confidence interval for each metric and strategy is also displayed. The confidence intervals were obtained by running the numerical experiment $100$ times in parallel.  
\fig{0.99}{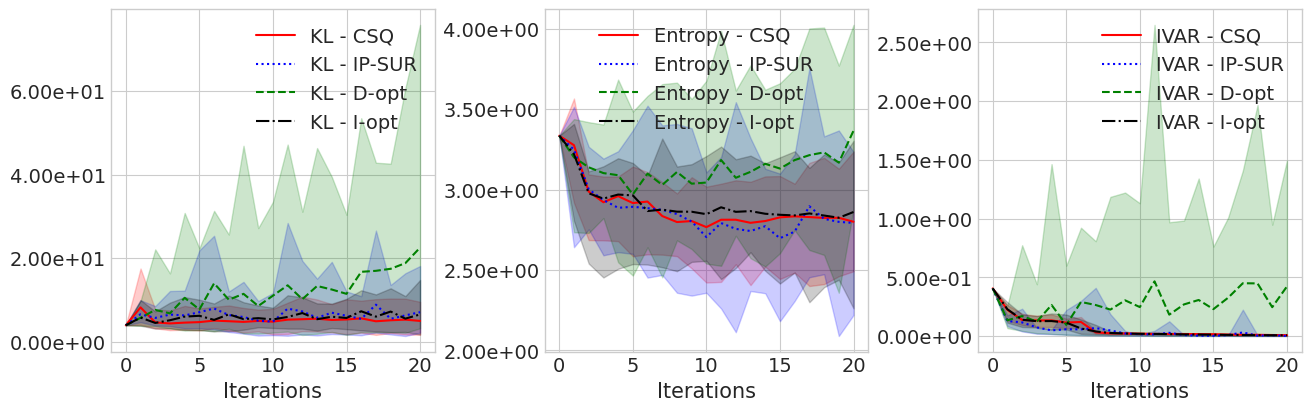}{Mean evolution and $95 \%$ confidence intervals of the performance metrics over $20$ iterations of the CSQ, IP-SUR,, D-optimal and I-optimal design strategies, replicated for $100$ experiments - Bimodal test case}{fig:bimodal_metrics}
\fig{0.99}{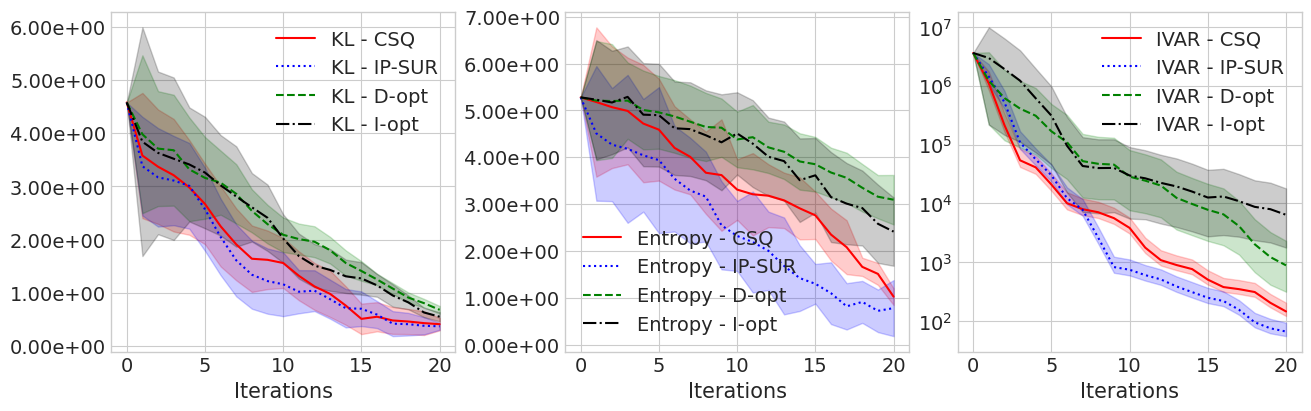}{Mean evolution and $95 \%$ confidence intervals of the performance metrics over $20$ iterations of the CSQ, IP-SUR, D-optimal and I-optimal design strategies, replicated for $100$ experiments - Neutronic test case}{fig:neutronic_metrics}
Both the CSQ and SUR strategies perform better than the D-optimal and I-optimal designs, which tend to target away from the posterior distribution. This increased performance is especially noticeable in the bimodal case and in the higher dimensional space of the neutronic test case.

Based on these results, one could argue that the CSQ strategy can be situationally better as it is easier to set up while providing similar performance in the end. However, two counter-arguments can be pointed out. First of all, the IP-SUR strategy does exhibit a guarantee for the convergence of the integrated variance, which offers a strong theoretical foundation. Besides, though the acquisition function in the IP-SUR strategy is more computationally intensive, the method does not rely on the prior setting of an arbitrary hyperparameter, while the CSQ design is based on the hyperparameter $h \in \R_+$ introduced in the definition of the bounding set $\mathcal{B}_n(h)$ in Equation \eqref{eq:bounding_set}. This hyperparameter quantifies how far away from the MAP the optimization problem can search. The choice of this hyperparameter can impact quite drastically the performance of the sequential design. A lower value leads to a more constraint set and thus an easier optimization problem, however, the new design points are confined to a smaller region of the parameter space. 
\fig{0.99}{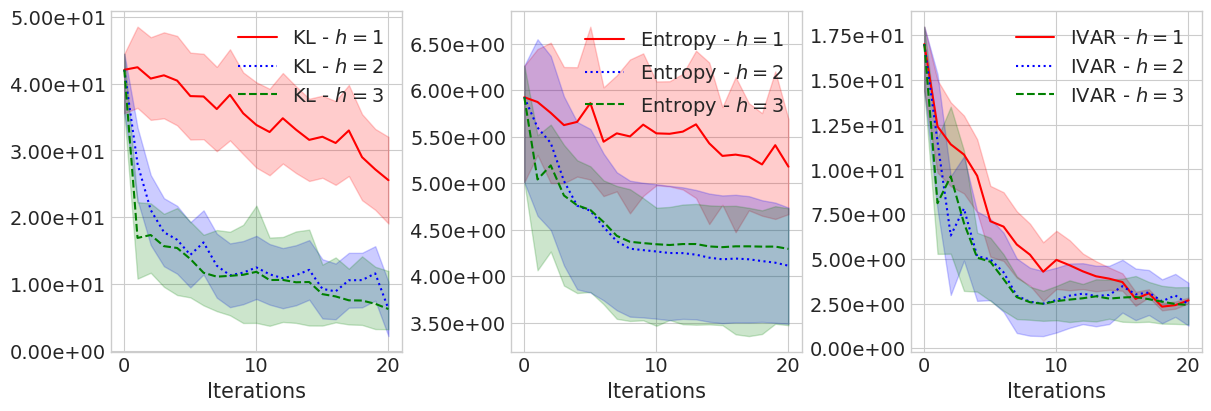}{Mean evolution and $95 \%$ confidence intervals of the performance metrics over $20$ iterations of the CSQ strategy for $h \in \{1, 2, 3\}$, replicated for $100$ experiments - Banana test case}{fig:hyperparam_metrics_banana}
\fig{0.99}{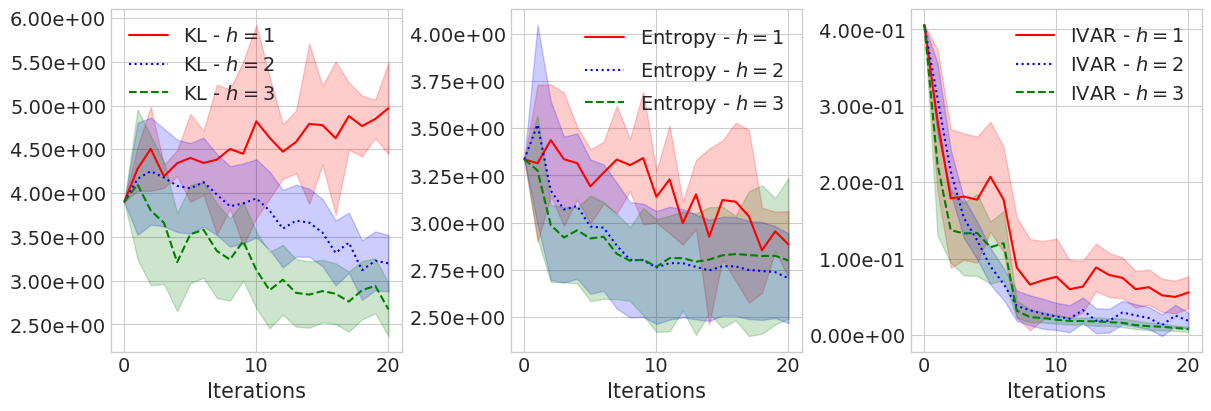}{Mean evolution and $95 \%$ confidence intervals of the performance metrics over $20$ iterations of the CSQ strategy for $h \in \{1, 2, 3\}$, replicated for $100$ experiments - Bimodal test case}{fig:hyperparam_metrics_bimodal}
To investigate the influence of $h$, the CSQ strategy is used once again to provide $20$ new design points for varying values of $h \in \left\{1, 2, 3 \right\}$. The performance metrics obtained for the banana-shaped and bimodal posterior distributions are shown respectively in Figures \ref{fig:hyperparam_metrics_banana} and \ref{fig:hyperparam_metrics_bimodal}.
One can see that for $h=1$, the design is significantly worsened. We recall that all the previous test cases were conducted with $h=3$ which provides the best results among the selected values. However, the optimal choice of $h$ is likely dependent on the application and cannot be found easily. For this reason, the IP-SUR strategy presented in this work seems superior as it does not carry the burden of the selection of a hyperparameter. 

Finally, let us compare the performance of our IP-SUR method with the design strategy introduced in \citep{sinsbeck2017sequential}. This sequential design strategy is based on the minimization of the Bayes risk with respect to a loss function measuring the variance of the likelihood estimate with the surrogate model. The same performance metrics are used for comparison.
\fig{0.99}{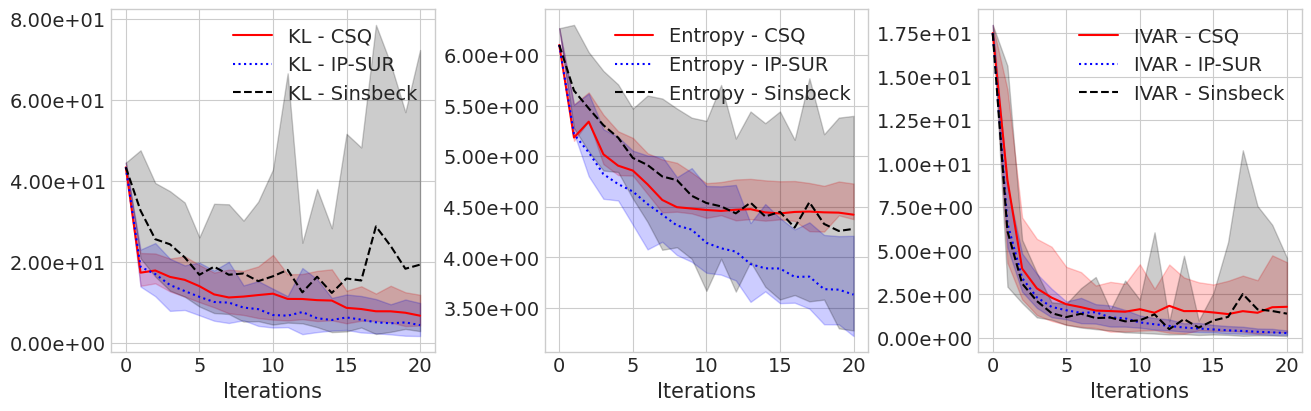}{Mean evolution and $95 \%$ confidence intervals of the performance metrics over $20$ iterations of the CSQ, IP-SUR and Sinsbeck \textit{et al.} strategy - Banana test case}{fig:banana_metrics_sinsbeck}
\fig{0.99}{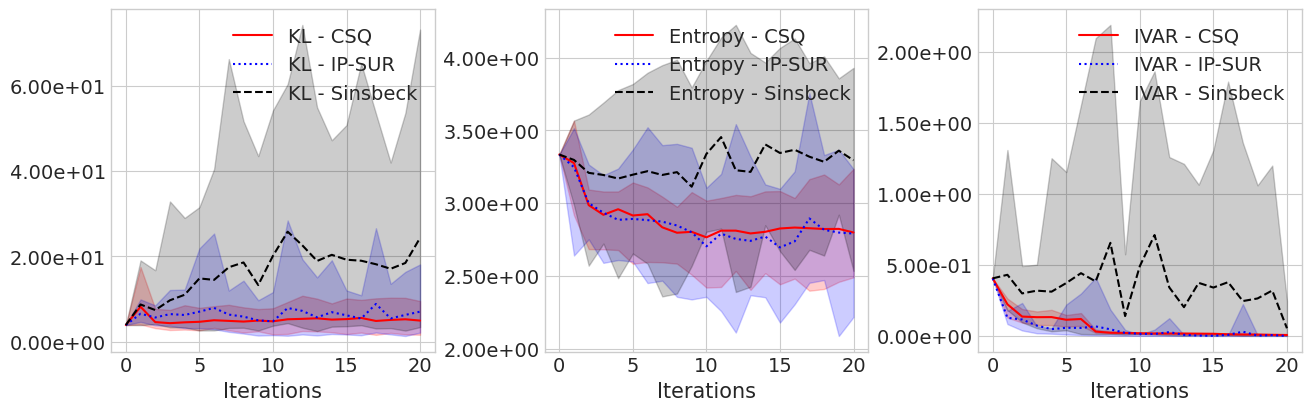}{Mean evolution and $95 \%$ confidence intervals of the performance metrics over $20$ iterations of the CSQ, IP-SUR and Sinsbeck \textit{et al.} strategy - Bimodal test case}{fig:bimodal_metrics_sinsbeck}
\fig{0.99}{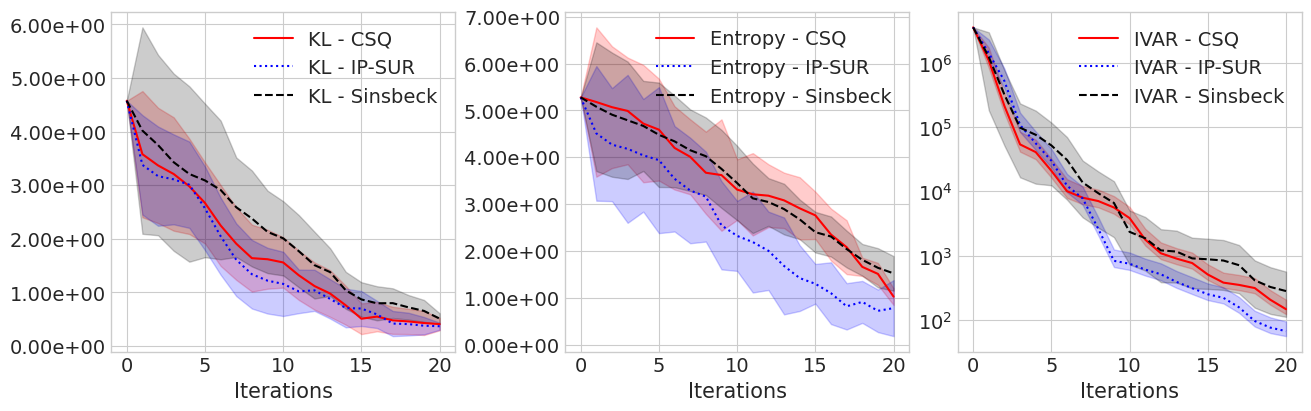}{Mean evolution and $95 \%$ confidence intervals of the performance metrics over $20$ iterations of the CSQ, IP-SUR and Sinsbeck \textit{et al.} strategy - Neutronic test case}{fig:neutronic_metrics_sinsbeck}
The evolutions of the performance metrics are displayed in Figures \ref{fig:banana_metrics_sinsbeck}, \ref{fig:bimodal_metrics_sinsbeck} and \ref{fig:neutronic_metrics_sinsbeck} for each test case. The two methods offer overall similar performance with regard to the performance metrics investigated for the banana and neutronic cases. For the bimodal case, however, our methods seem more reliable. Besides, the IP-SUR is backed by a convergence guarantee which is not the case for the method of Sinsbeck \textit{et al.}. For the latter, the supermartingale property seems unreachable.

\section*{Conclusion}
This work presents two new sequential design strategies to build efficient Gaussian process surrogate models in Bayesian inverse problems. These strategies are especially important for cases where the posterior distribution in the inverse problem has thin support or is high-dimensional, in which case space-filling designs are not as competitive. The IP-SUR strategy introduced in this work is shown to be tractable and is supported by a theoretical guarantee of almost sure convergence of the weighted integrated mean square prediction error to zero. This method is compared to a simpler CSQ strategy which is adapted from D-optimal designs and to a strategy based on the minimization of the Bayes risk with respect to the variance of the likelihood estimate. While both methods perform better than D-optimal and I-optimal strategies, the IP-SUR method seems to provide better performance than CSQ for higher dimensions while not relying on the choice of a hyperparameter, all the while being grounded on strong theoretical foundations. It is also comparable to the Bayes risk minimization for all test cases and even superior for the bimodal test case. The latter strategy also does not display a convergence guarantee.
The IP-SUR criterion developed in this work can be related to the traditional IMSPE criterion by a tempered variant of the design strategy \citep{neal2001annealed, hu2002weighted, del2006sequential} in which we could introduce a tempering parameter $\beta \in [0, 1]$ and define a new design strategy based on the functional $\mathcal{H}_{\beta}(\nu) = \frac{1}{Z_{\nu, \beta}} \int_{\mathcal{X}} |\Cnu(x_*)| \left(L_{\nu}(\mathbf{y} | x_*) \right)^{\beta} p(x_*) dx_*$ with $Z_{\nu, \beta} = \int_\mathcal{X} \left(L_{\nu}(\mathbf{y} | x_*) \right)^{\beta} p(x_*) dx_*$. The IMSPE criterion can be obtained for $\beta = 0$ while the IP-SUR strategy is obtained for $\beta = 1$. The almost sure convergence can be verified with the same approach as for $\mathcal{H}$. 

The extension of this design strategy is yet to be explored for other types of surrogate models. Gaussian process surrogate models have the advantage of keeping the SUR criterion tractable,  though one could consider a variant of this method with approximation methods to evaluate the SUR criterion, with the risk of losing the guarantee on the convergence. The robustness of such an approach would also have to be investigated.

\newpage
\bibliographystyle{apalike}                  

\bibliography{bibliography}
\newpage
\appendix
\section{Additional proofs}
\subsection{Proof of Proposition \ref{prop:sur_criteria}}\label{app:proof_sur_criteria}

\renewcommand{\thesection}{\Alph{section}}
\renewcommand{\theproposition}{\thesection.\arabic{proposition}}

\setcounter{section}{1}

\begin{proof}
    The quantity of interest is $F_n(x) = \mathbb{E}_{z_{n+1}} \left[D_{n+1}(x, z_{n+1}) \right]$ where $z_{n+1} = z^{(n)}(x) \sim \mathcal{N}\left(m_n(x), \Cn(x) \right)$.
    Writing explicitly the expectation and inverting the integrals we get:
    \begin{align*}
       F_n(x) &= \int_\mathcal{X} \int_{\R^d}  L_{n+1}(\mathbf{y}| x_*, x, z_{n+1}) p(x_*) | \mathbf{C}_{n+1}(x_* | x) | p_n(z_{n+1} | x) dz_{n+1} dx_*
    \end{align*}
    where $p_n(z_{n+1} | x)$ is the density of the next model output $z_{n+1}$ given $x \in \mathcal{X}$, which is given by the GP predictive distribution $\mathcal{N}\left(m_n(x), \mathbf{C}_n(x) \right)$. This density does not depend on $x_*$ and thus $p_n(z_{n+1} | x) = p_n(z_{n+1} | x, x_*)$. With that, we have:
    \begin{align*}
        F_n(x) &= \int_\mathcal{X} | \mathbf{C}_{n+1}(x_* | x) | p(x_*) \int_{\R^d} p_n(\mathbf
        {y}, z_{n+1} | x, x_*) dz_{n+1} dx_* \\
        &= \int_\mathcal{X} | \mathbf{C}_{n+1}(x_* | x) | p(x_*) L_n(\mathbf
        {y} | x, x_*)dx_*.
    \end{align*}
    where $p_n(\mathbf{y}, z_{n+1} | x, x_*)$ is the conditional probability density of $\mathbf{y}$ and $z_{n+1}$ given $(x, x_*)$.
    This concludes the proof since $L_n(\mathbf{y} | x, x_*)$ does not depend on $x$.
\end{proof}

\subsection{Proof of Lemma  \ref{prop:smp_for_D} (supermartingale property for \texorpdfstring{$\mathcal{D}$}{D})}\label{app:proof_smp}



\renewcommand{\thesection}{\Alph{section}}
\renewcommand{\theproposition}{\thesection.\arabic{lemma}}
\setcounter{section}{1}
\begin{proof}
Let $x \in \mathcal{X}$. We would like to show that $F_n(x) \leq D_n$ for all $n \geq 1$, where $F_n(x) = \mathbb{E}_{n, x}\left[\mathcal{D}(\nu_n) \right] = \mathbb{E}_{z_{n+1}} \left[D_{n+1}(x, z_{n+1}) \right]$. Let $x \in \mathcal{X}$ and $n \geq 1$. We have:
\begin{equation}\label{eq:f_n}
    F_n(x) = \int_\mathcal{X} |\mathbf{C}_{n+1}(x_* | x) |  L_n(\mathbf{y} | x_*) p(x_*) dx_*.
\end{equation}
The supermartingale property derives naturally from this expression since we have $|\mathbf{C}_{n+1}(x_* | x) | \leq |\Cn(x_*)|$ for all $x \in \mathcal{X}$.
\end{proof}

\subsection{Proof of Lemma \ref{lemma}}\label{app:proof_lemma}



\renewcommand{\thesection}{\Alph{section}}
\renewcommand{\theproposition}{\thesection.\arabic{lemma}}
\setcounter{section}{1}
\begin{proof}
    Let us first prove the convergence of the mean functions $(m_n)_{n \geq 1}$. From proposition 2.9 in \cite{bect2016supermartingale}, for any sequential design and for any Gaussian process, the probability distribution of the GP given $\mathcal{F}_n = \sigma(x_1, z_1, ..., x_n, z_n)$, which is denoted by $\nu_n$, converges almost surely to a limit Gaussian measure $\nu_\infty \in \mathbb{M}$. Since the convergence of Gaussian measures is defined as the uniform convergence of the mean functions $(m_n)_{n \geq 1}$ and covariance functions $(\mathbf{C}_n)_{n \geq 1}$ we can then define $m_\infty = \lim_{n \to +\infty} m_n$ which is the mean function of the GP whose probability distribution is $\nu_\infty$. Similarly we define $\mathbf{C}_\infty = \lim_{n \to +\infty} \mathbf{C}_n$. Furthermore, since $m_n$ and $\mathbf{C}_n$ are continuous and $(m_n)_{n \geq 1}$ (resp. $(\mathbf{C}_n)_{n \geq 1}$) converges uniformly to $m_\infty$ (resp. $\mathbf{C}_\infty$), then $m_\infty$(resp. $\mathbf{C}_\infty$) is continuous.
    
    Let $x_* \in \mathcal{X}$. We need to show that $\Ctot(x_*) \xrightarrow[n \to +\infty]{a. s.} \boldsymbol\Sigma_{\infty}(x_*) = \mathbf{C}_\infty(x_*) \otimes \mathcal{U}_N + \Cobs \otimes \mathcal{I}_N$. By continuity of all the other matrix operations, we have $L_n(\mathbf{y} | x_*) \xrightarrow[n \to +\infty]{a. s.} L_\infty(\mathbf{y} | x_*)$ with:
    \begin{equation*}
        L_\infty(\mathbf{y} | x_*) = ((2 \pi)^{Nd} | \boldsymbol\Sigma_\infty(x_*) | )^{-1 / 2} \exp \left[- \frac{1}{2}\left\| \mathbf{y - m_\infty}(x_*)\right\|^2_{\boldsymbol\Sigma_\infty(x_*)} \right].
    \end{equation*} 
    To conclude, we just need to notice that $L_\infty$ is continuous (with respect to $x_*$) since $m_\infty$ and $\mathbf{C}_\infty$ are continuous. It is thus bounded on the compact set $\mathcal{X}$ and we can make use of the dominated convergence theorem to conclude:
    \begin{equation*}
        \lim_{n \to +\infty} \int_{\mathcal{X}} L_n(\mathbf{y} | x_*) p(x_*) dx_* = \int_{\mathcal{X}} \lim_{n \to +\infty} 
        L_n(\mathbf{y} | x_*) p(x_*) dx_* = \int_{\mathcal{X}} L_\infty(\mathbf{y} | x_*) p(x_*) dx_* = Z_\infty.
    \end{equation*}
    Besides, since $L_{\infty}(\mathbf{y} | x_*) p(x_*) > 0$ for all $x_* \in \mathcal{X}$ and $\mathcal{X}$ is a compact, $0 < Z_{\infty} < + \infty$ almost surely. 
\end{proof}

\subsection{Proof of Theorem \ref{th:conv_H} (almost sure convergence of \texorpdfstring{$H_n$}{Hn})}\label{app:conv_proof}



\renewcommand{\thesection}{\Alph{section}}
\renewcommand{\theproposition}{\thesection.\arabic{lemma}}
\setcounter{section}{1}

\begin{proof}
This proof is divided into two parts. First of all, we show that $D_n = \mathcal{D}(\nu_n) \xlongrightarrow[n \to +\infty]{a.s.} 0$ using the supermartingale property and the convergence theorem from \cite{bect2016supermartingale}. Then, we show the convergence for $H_n$ using Lemma \ref{lemma}. 

Let us first show that $D_n\xlongrightarrow[n \to +\infty]{a.s.} 0$. Consider a Gaussian process $z$ and a SUR  sequential design $(x_n)_{n \geq 1}$ given by the strategy \eqref{eq:sur_strategy} defined for the functional $\mathcal{D}$. Let us verify the conditions of Theorem \ref{conv_theorem} for the functional $\mathcal{D}$. We introduce the functional $\mathcal{G}$ which is defined for $\nu \in \mathbb{M}$ by:
\begin{equation*}
    \mathcal{G}(\nu) = \sup_{x \in \mathcal{X}}\left(\mathcal{D}(\nu) - \mathbb{E}_{z(x)}\left[\mathcal{D}(\nu | (x, z)) \right] \right).
\end{equation*} 
If $\nu_n$ is the probability of $z$ given $\mathcal{F}_n = \sigma(x_1, z_1, .., x_n, z_n)$, then there exists an $\mathcal{F}_{\infty}$-measurable random element $\nu_\infty \in \mathbb{M}$ such that $\nu_n \xlongrightarrow[n \to +\infty]{a.s.} \nu_\infty$, with $\mathcal{F}_{\infty} = \sigma\left( \bigcup_{n\geq 1} \mathcal{F}_n \right)$, and such that $ \nu_\infty$ is the conditional probability of $z$ given $\mathcal{F}_{\infty}$. The associated mean and covariance functions are introduced as $m_\infty$ and $\mathbf{C}_\infty$. Since the functions $m_n$ are continuous and the sequence $(m_n)_{n \geq 1}$ converge uniformly toward the limit mean function $m_\infty$, then $m_\infty$ is also continuous. The same reasoning holds for $\mathbf{C}_\infty$. 

From here, the convergence of $\mathcal{D}(\nu_n) \xrightarrow[n \to +\infty]{a.s.} \mathcal{D}( \nu_\infty)$ and $\mathcal{G}(\nu_n) \xrightarrow[n \to +\infty]{a.s.} \mathcal{G}( \nu_\infty)$ is obtained by dominated convergence on the compact $\mathcal{X}$ and by continuity of the mean and variance of the Gaussian posterior measure. 

Now, let us show that $\mathbb{Z}_{\mathcal{D}} = \mathbb{Z}_{\mathcal{G}}$, where $\mathbb{Z}_{\mathcal{D}}$ and $\mathbb{Z}_{\mathcal{G}}$ are the sets of zeros of the functionals $\mathcal{D}$ and $\mathcal{G}$. The first inclusion $\mathbb{Z}_{\mathcal{D}} \subset \mathbb{Z}_{\mathcal{G}}$ is trivial since $0 \leq \mathcal{G}(\nu) \leq \mathcal{D}(\nu)$ for all $\nu \in \mathbb{M}$. Now, let $\nu \in \mathbb{Z}_{\mathcal{G}}$. It is a Gaussian measure associated to a GP $z$. We want to show that $\mathcal{D}(\nu) = 0$. We will prove this result for the posterior measure $\nu_n$ for all $n \geq 0$ to be able to re-use the previous notations, though we are only interested in the case $n = 0$. In particular $\nu = \nu_0$.

Let us introduce $F_n(x) = \mathbb{E}_{n, x}\left[\mathcal{D}(\nu_{n+1}) \right] = \mathbb{E}_{z(x)}\left[\mathcal{D}\left(\nu_n | (x, z) \right) \right]$. From this, $\mathcal{G}(\nu_n) = D_n - \inf_{x\in \mathcal{X}} F_n(x) = 0$. Then, the supermartingale property of $\mathcal{D}$ yields that for all $x \in \mathcal{X}$:
\begin{equation*}
    0 \leq D_n - F_n(x) \leq D_n - \inf_{x\in \mathcal{X}} F_n(x) = 0
\end{equation*}
thus $D_n - F_n(x) = 0$ for all $x \in \mathcal{X}$. 

Using Equations \eqref{eq:f_n} and \eqref{eq:cov_update}, if $D_n - F_n(x) = 0$, then for almost all $x_* \in \mathcal{X}$, we have $\Cn(x_*, x) = 0$, since $L_n(\mathbf{y} | x_*) > 0$ for all $x_* \in \mathcal{X}$. In other words, the set $\mathcal{C}_x = \left\{ x_* \in \mathcal{X} | \Cn(x_*, x)\neq 0 \right\}$ has Lebesgue measure zero, for all $x \in \mathcal{X}$. \\
Let us proceed by contradiction. We suppose that $D_n \neq 0$. Thus there exists $\mathcal{X}_1 \subset \mathcal{X}$ such that $\mu(\mathcal{X}_1) > 0$ and for $x_* \in \mathcal{X}_1$, we have $\Cn(x_*) > 0$. Let $x \in \mathcal{X}_1$. Taking $x_* = x$, we have $\Cn(x_*, x) = \Cn(x) > 0$. Thus by continuity of $\Cn$, there exists an open set $\mathcal{X}_2 \subset \mathcal{X}$ such that $\mu(\mathcal{X}_2) > 0$ and for all $x_* \in \mathcal{X}_2$, $\Cn(x_*, x) > 0$. We have our contradiction since $\mathcal{X}_2 \subset \mathcal{C}_x$ and $\mu(\mathcal{C}_x) = 0$. Therefore we conclude that $\mathbb{Z}_{\mathcal{D}} = \mathbb{Z}_{\mathcal{G}}$. \\
All the assumptions of the Theorem \ref{conv_theorem} are verified, and the almost sure convergence of $D_n$ is proven:
\begin{equation*}
    D_n\xlongrightarrow[n \to +\infty]{a.s.} 0.
\end{equation*}
Now consider $(Z_n)_{n \geq 1}$ the sequence of normalizing constants, defined by $Z_n = \int_\mathcal{X} L_n(\mathbf{y} | x_*) p(x_*) dx_*$. From Lemma \ref{lemma} and using the same notations, $\lim_{n \to +\infty} Z_n= Z_\infty = \int_{\mathcal{X}} L_\infty(\mathbf{y} |x_*) p(x_*) dx_*$ which is positive almost surely. Since $D_n\xlongrightarrow[n \to +\infty]{a.s.} 0$ and $H_n = \frac{D_n}{Z_n}$, we can conclude that:
\begin{equation*}
    H_n\xlongrightarrow[n \to +\infty]{a.s.} 0.
\end{equation*}
\end{proof}

\section{Point model approximation in neutron noise analysis}\label{app:neutronic}
Neutron noise analysis describes a set of techniques that study the temporal fluctuations of neutron detector responses. For this particular work, the goal is to identify a fissile nuclear material based on measurements of temporal correlations between neutrons created by induced fissions inside the unknown material. From such noisy observations, an inverse problem is solved to identify the unknown material. In its simplest form, the link between material and observations is given by the so-called point model approximation, which is detailed hereafter.  

The material is identified by a set of parameters $x$. In the simplest formulation of the point model, four parameters are considered such that $x \in \mathcal{X} \subset \R^4$.
\begin{itemize}
    \item $0 < k_p < 1$ is the prompt multiplication factor.
    \item $\varepsilon_F$ is the ratio of detected neutrons over the number of induced fissions in the material. 
    \item $S$ is the source intensity in neutrons per second.
    \item $x_s$ is the ratio of source neutrons produced by spontaneous fissions, over the total number of source neutrons. 
\end{itemize}
The source neutrons are either created by nuclear reactions modeled by Poisson point processes, or by spontaneous fissions which are modeled by compound Poisson processes. 

The statistical model includes three different observations $y \in  \R^3$, which are recorded for each numerical simulation (or practical experiments).
\begin{itemize}
    \item $R$ is the average detection rate of neutrons in the detector. 
    \item $Y_{\infty}$ is the second order asymptotic Feynman moment. 
    \item $X_{\infty}$ is the third order asymptotic Feynman moment.
\end{itemize}
The interpretation of the two quantities $Y_{\infty}$ and $X_{\infty}$ is not detailed here. To simplify, they can be viewed as the binomial moments of order $2$ and $3$ of the number of detected neutrons in some given time window of size $T$, where $T$ is taken to be much larger than the average lifetime of a fission chain. For more details, the authors refer to \citep{pazsit2007neutron, furuhashi1968third} and \citep{feynman1956dispersion}.

In the point model framework, strong assumptions are made to provide an analytical link between inputs $x$ and outputs $y$. The material is assumed uniform, homogeneous, and infinite. All the neutrons have the same energy, and the only nuclear reactions considered are neutron captures and fissions. 

In this context, analytical relations can be derived. The reactivity $\rho = \frac{k-1}{k} < 0$ is introduced for concision.
\begin{equation*}
    R = - \frac{\varepsilon_F S \overline{\nu_s}}{\rho \overline{\nu}(\overline{\nu_s} + x_s - \overline{\nu_s} x_s)} 
\end{equation*}
\begin{equation*}
   Y_{\infty} = \frac{\varepsilon_F D_2}{\rho^2}\left(1 - x_s \rho \frac{\overline{\nu_s} D_{2s}}{\overline{\nu} D_2} \right)
\end{equation*}
\begin{equation*}
    X_{\infty} = 3\left( \frac{\varepsilon_F D_2}{\rho^2} \right)^2 \left(1 - x_s \rho \frac{\overline{\nu_s} D_{2s}}{\overline{\nu} D_2} \right) - \frac{\varepsilon_F^2 D_3}{\rho^3}\left(1 - x_s \rho \frac{\overline{\nu_s}^2 D_{3s}}{\overline{\nu}^2 D_3} \right).
\end{equation*}
In these relations, $\overline{\nu}$, $D_2$, $D_3$ are nuclear data quantifying the multiplicity distribution of the neutrons created by induced fissions. Similarly, $\overline{\nu_s}$, $D_{2s}$, $D_{3s}$ describe the multiplicity of the neutrons created by spontaneous fissions. These quantities are considered known values in this work. 

\end{document}